\documentclass[english,pra,preprint,amsmath,amssymb,aps,longbibliography,titlepage,superscriptaddress,eprint,showkeys,linenumbers]{revtex4-2}
\usepackage[T1]{fontenc}
\usepackage[latin9]{inputenc}
\usepackage{mathtools}
\usepackage{amsmath}
\usepackage{amsthm}
\usepackage{amssymb}
\usepackage{graphicx}

\makeatletter
\theoremstyle{plain}
\newtheorem{thm}{\protect\theoremname}
\theoremstyle{plain}
\newtheorem{lem}{\protect\lemmaname}

\makeatletter

\usepackage{dcolumn}% Align table columns on decimal point
\usepackage{bm}% bold math
\usepackage[justification=raggedright]{caption}
\usepackage{subcaption}
\usepackage{physics}
\usepackage{algorithm}
\usepackage{algorithmic}
\usepackage{comment}
\usepackage{xcolor} % use color in text

\usepackage[hypertexnames=false]{hyperref}
\hypersetup{
	breaklinks = true,
    colorlinks = true,
    citecolor = {blue},
	urlcolor = {blue},
	linkcolor = {blue}
}

\makeatother

\usepackage{babel}
\providecommand{\lemmaname}{Lemma}
\providecommand{\theoremname}{Theorem}

\begin{document}
\nolinenumbers
\title{An Explicitly Solvable Energy-Conserving Algorithm for Pitch-Angle Scattering in Magnetized Plasmas}
\author{Yichen Fu}
\affiliation{Princeton Plasma Physics Laboratory, Princeton University}
\affiliation{Department of Astrophysical Sciences, Princeton University}
\author{Xin Zhang}
\affiliation{Princeton Plasma Physics Laboratory, Princeton University}
\affiliation{Department of Astrophysical Sciences, Princeton University}
\author{Hong Qin}
\email{hongqin@princeton.edu}

\affiliation{Princeton Plasma Physics Laboratory, Princeton University}
\affiliation{Department of Astrophysical Sciences, Princeton University}
\begin{abstract}
We develop an Explicitly Solvable Energy-Conserving (ESEC) algorithm for the Stochastic Differential Equation (SDE) describing the pitch-angle scattering process in magnetized plasmas. The Cayley transform is used to calculate both the deterministic gyromotion and stochastic scattering, affording the algorithm to be explicitly solvable and exactly energy conserving. An unusual property of the SDE for pitch-angle scattering is that its coefficients diverge at the zero velocity and do not satisfy the global Lipschitz condition. Consequently, when standard numerical methods, such as the Euler-Maruyama (EM), are applied, numerical convergence is difficult to establish. For the proposed ESEC algorithm, its energy-preserving property enables us to overcome this obstacle. We rigorously prove that the ESEC algorithm is order 1/2 strongly convergent. This result is confirmed by detailed numerical studies. For the case of pitch-angle scattering in a magnetized plasma with a constant magnetic field, the numerical solution is benchmarked against the analytical solution, and excellent agreements are found.  
\end{abstract}
\maketitle

\section{Introduction\label{sec:Introduction}}

Coulomb collision between charged particles is an important physical process in plasmas. For most systems, it is dominated by the accumulative effects of small-angle scatterings. In the Fokker-Planck (FP) equation for particle distribution functions, it is described by the Landau collision operator \citep{landau1936kinetic}, \begin{linenomath} 
\begin{align}
C_{ab}[f_{a},f_{b}]=-\dfrac{q_{a}^{2}q_{b}^{2}\ln\Lambda}{8\pi\epsilon_{0}^{2}m_{a}}\dfrac{\partial}{\partial\boldsymbol{v}}\cdot\int\left(\dfrac{u^{2}\boldsymbol{I}-\boldsymbol{uu}}{u^{3}}\right)\left[\dfrac{f_{a}(\boldsymbol{v})}{m_{b}}\dfrac{\partial f_{b}(\boldsymbol{v}')}{\partial\boldsymbol{v}'}-\dfrac{f_{b}(\boldsymbol{v}')}{m_{a}}\dfrac{\partial f_{a}(\boldsymbol{v})}{\partial\boldsymbol{v}}\right]\mathrm{d}^{3}\boldsymbol{v}',\label{eq:landau}
\end{align}
\end{linenomath} where $C_{ab}$ denotes the collision experienced by particle species $a$ colliding with particle species $b$, $q_{a(b)},m_{a(b)}$ are the charge and mass of species $a(b)$, $\ln\Lambda$ is the Coulomb logarithm, $\boldsymbol{u}=\boldsymbol{v}-\boldsymbol{v}'$, and $u=|\boldsymbol{u}|$ is the Euclidean norm of $\boldsymbol{u}$. Due to the nonlinearity in distribution functions, FP equations with the Landau collision operators are difficult to solve either analytically or numerically. Fortunately, in many scenarios, simplifications are possible when the system has a well-defined separation of ordering. For the case of a low-density species propagating in a Maxwellian background species, the distribution of the background species $f_{b}(\boldsymbol{v})$ can be taken to be time-independent, and the collision operator $C_{ab}$ can be written as a linear function of the distribution of the propagating species $f_{a}$. When the mass ratio $m_{a}/m_{b}$ is small, the lead order of $C_{ab}$ reduces to the Lorentz collision operator modeling the pitch-angle scattering process, 
\begin{equation}
\mathcal{L}[f_{a}]=\dfrac{1}{2}\dfrac{\partial}{\partial\boldsymbol{v}}\cdot\left[\left(\dfrac{v^{2}\boldsymbol{I}-\boldsymbol{vv}}{v^{3}}\right)\cdot\dfrac{\partial f_{a}}{\partial\boldsymbol{v}}\right],
\end{equation}
where time has been normalized by the collision frequency $\nu_{ei}$ , and $\boldsymbol{v}$ by the thermal velocity. Pitch-angle scattering is an important physical process in plasma physics for several reasons. It is the dominant process for tokamak current drive \citep{Fisch1978,Fisch87,Karney86} and plays an important role in the runaway electron dynamics \citep{Rosenbluth97,Karney86,Boozer2015}. For first-principles-based simulation studies, a reliable algorithm for pitch-angle scattering is needed. In addition, equations of a similar structure appear in the models of stochastic Landau-Lifshitz-Gilbert dynamics for magnetization \citep{d2006midpoint}, where similar algorithmic issues and challenges exist. Th algorithm that we develop in the present study can be applied to simulate these processes in condensed matter physics as well.

There are two major types of numerical methods to simulate the Coulomb collision, the continuum methods \citep{chacon2000implicit,yoon2014fokker,kraus2017metriplectic,Hirvijoki2018} and Monte-Carlo methods \citep{takizuka1977binary,nanbu1997theory,manheimer1997langevin,cadjan1999langevin,sherlock2008monte,lemons2009small,cohen2010time,dimits2013higher,PhysRevE.102.033302}. One commonly used Monte-Carlo method is directly calculating the binary collisions between random pairs of particles \citep{takizuka1977binary,nanbu1997theory}. The particles are scattered elastically within spatially localized cells. But this method may not be computational efficient when the number of particles is large \citep{verboncoeur2005particle}. Another approach is to solve the corresponding Stochastic Differential Equations (SDE), i.e., the Langevin equations, of the FP equation \citep{manheimer1997langevin,cadjan1999langevin,sherlock2008monte,lemons2009small,cohen2010time,dimits2013higher,PhysRevE.102.033302}. In this approach, the SDE corresponding to the FP equation needs to be derived.

In the present study, we are concerned with the physics of pitch-angle scattering of electrons in a magnetized plasma, which can be equivalently described by the FP equation, 
\begin{equation}
\dfrac{\partial f_{e}(\boldsymbol{v},t)}{\partial t}=-\big[\boldsymbol{v}\times\boldsymbol{B}(t)\big]\cdot\dfrac{\partial f_{e}}{\partial\boldsymbol{v}}+\mathcal{L}[f_{e}(\boldsymbol{v},t)],\label{FP}
\end{equation}
or the corresponding Ito SDE, \begin{linenomath} 
\begin{align}
\mathrm{d}\boldsymbol{v}=\left(\boldsymbol{v}\times\boldsymbol{B}(t)-D(v)\dfrac{\boldsymbol{v}}{v^{2}}\right)\mathrm{d}t+\sqrt{D(v)}\left(\boldsymbol{I}-\dfrac{\boldsymbol{vv}}{v^{2}}\right)\cdot\mathrm{d}\boldsymbol{W}.\label{eq:Ito_SDE}
\end{align}
\end{linenomath} In Eq.\,(\ref{eq:Ito_SDE}), $\boldsymbol{W}(t)$ is the three dimensional Wiener process, $D(v)=1/v$ is the pitch-angle diffusion coefficient, and $\boldsymbol{B}$ has been normalized by $e$, $m_{e},$ and $\nu_{ei}$. We will focus on developing an effective algorithm for solving Eq.\,(\ref{eq:Ito_SDE}).

An unusual property of Eq.\,(\ref{eq:Ito_SDE}) is that its coefficients are not globally Lipschitz continuous and diverge at $\boldsymbol{v}=0$. However, for a given initial condition $\boldsymbol{v}(t_{0})=\boldsymbol{v}_{0}\neq0$, it can be easily proved that {[}see Eq.\,(\ref{eq:dv=00003D00003D0}){]} the exact solution of Eq.\,(\ref{eq:Ito_SDE}) preserves the norm of $\boldsymbol{v}$, i.e., $v=v_{0}\neq0$. Consequently, the singularity at $\boldsymbol{v}=0$ is irrelevant for exact solutions.

One can apply either the standard Euler-Maruyama (EM) algorithm or a higher-order method, such as the Milstein algorithm \citep{kloeden2013numerical}, to Eq.\,(\ref{eq:Ito_SDE}). However, in Cartesian coordinates, these SDE methods do not typically conserve particle energy due to the truncation error, accumulation of which can lead to numerical instabilities \citep{lemons1995noise} and divergent numerical solutions. Such accumulation in errors in particle energy may have disastrous consequences in physics implications.

For pitch-angle scattering, a common method to ensure energy conservation in SDE-based simulations is to cast Eq.~(\ref{eq:Ito_SDE}) in the spherical coordinates $(v,\varphi,\theta)$, where $v$ is the velocity magnitude, $\varphi$ the azimuthal angle, and $\theta$ the polar (pitch) angle \citep{lemons2009small}. In practice, the $(v,\varphi,\mu=\cos\theta)$ coordinates are often used instead \citep{rosenbluth1957fokker,boozer1981monte,dimits2013higher}. To simplify the discussion, let's consider the unmagnetized case here. In the $(v,\varphi,\mu)$ coordinates, Eq.~(\ref{eq:Ito_SDE}) is \begin{linenomath} 
\begin{align}
\mathrm{d}\mu & =-D_{a}(v)\mu\mathrm{d}t+\sqrt{D_{a}(v)(1-\mu^{2})}\mathrm{d}W_{\mu},\label{eq:Ito_SDE_spherical_coordinate}\\
\mathrm{d}\varphi & =\sqrt{\frac{D_{a}(v)}{1-\mu^{2}}}\mathrm{d}W_{\varphi},\label{eq:Ito_psi}\\
\mathrm{d}v & =0,\label{eq:Ito_sph_r}
\end{align}
\end{linenomath} where $D_{a}=1/v^{3}$ is the angular diffusion coefficient, and $W_{\mu}$ and $W_{\varphi}$ are independent Wiener processes. It is obvious that only Eqs.\,(\ref{eq:Ito_SDE_spherical_coordinate}) and (\ref{eq:Ito_psi}) need to be solved, and the energy is automatically conserved. Here, $\mu$ should be bounded in the interval of $[-1,1]$, but standard numerical schemes, such as the Euler-Maruyama or the Milstein scheme, do not preserve this bound \citep{lemons2009small}. When $|\mu|>1$, the coefficient of Eq.\,(\ref{eq:Ito_SDE_spherical_coordinate}) becomes imaginary and unphysical, and empirical modifications are required \citep{rosin2014multilevel}. Furthermore, when $\mu$ is close to $1$, which does happen regularly, the coefficient of Eq.\,(\ref{eq:Ito_psi}) diverges, and the algorithm cannot correctly calculate the dynamics of $\varphi$. In a magnetized plasma with a constant, strong magnetic field, one can ignore this failure by choosing $\varphi$ to be the gyrophase, and ignoring the dynamics of $\varphi$ based on the gyrokinetic ordering. However, this strategy is invalid if the magnetic field is weak or depends on space and/or time, or if we are interested in the physics depending the gyrophase, such as the cyclotron waves.

Beside this shortcoming, there is another mathematical difficulty associated with the SDE for the pitch-angle scattering. In Eqs.~(\ref{eq:Ito_SDE}), (\ref{eq:Ito_SDE_spherical_coordinate}) and (\ref{eq:Ito_psi}), both $D(v)$ and $1/\sqrt{1-\mu^{2}}$ are not globally Lipschitz continuous and diverge in the neighborhood of ${v}=0$ and $|\mu|=1$, respectively. However, the classical proof of convergence for the EM scheme requires that both the drift and the diffusion coefficients are globally Lipschitz continuous and have linear growth bounds \citep{kloeden2013numerical,milstein2013stochastic}. Many studies on the convergence of numerical schemes, especially for the EM method, have been carried out for SDEs with non-global Lipschitz coefficients \citep{higham2002strong,hutzenthaler2011strong}. In these studies, the coefficients of the SDEs are not Lipschitz continuous at infinity, which is different from the case of the pitch-angle scattering. We are not aware of any existing proof of convergence for these numerical schemes when applied to Eqs.\,(\ref{eq:Ito_SDE}), (\ref{eq:Ito_SDE_spherical_coordinate}), and (\ref{eq:Ito_psi}). In practical simulations, if we are only interested in the statistical properties of the ensemble of particles when using the Monte-Carlo method, this problem could be easily overlooked. Upon closer examinations of individual sample paths, however, one may find that a small fraction of them diverges due to $v$ drifting towards zero. This will be demonstrated numerically in Sec.\,\ref{sec:numerical_experiments}. Ad hoc remedies include regularizing the coefficient when $v$ is smaller than a critical $v_{c}$ \citep{rosin2014multilevel}, the value of which can only be determined empirically without a systematic approach. Other times the stochastic term is simply ignored based on other physical considerations \citep{lemons2009small} when $v$ is small.

Recently, we proposed an Explicitly Solvable Energy-Conserving scheme (ESEC) to solve the pitch-angle scattering SDE in unmagnetized plasmas using Cartesian coordinates \citep{PhysRevE.102.033302}. The ESEC scheme applies the Cayley transform \citep{qin2013boris,Shang2013private} to rotate the velocity, which guarantees the exact conservation of energy. It has been demonstrated that the ESEC method has an order $1/2$ global strong error, the same as the EM scheme when applied to SDEs satisfying the global Lipschitz condition. Since the calculation is energy-preserving and performed in Cartesian coordinates, the algorithm is not subject to the problems of imaginary and divergent coefficients associated with spherical coordinates. In the present study, we generalize the scheme proposed in Ref.\,\citep{PhysRevE.102.033302} to Eq.\,(\ref{eq:Ito_SDE}) in a magnetized plasma, and rigorously prove the strong convergence of the algorithm. We emphasize again that because the coefficients of Eq.\,(\ref{eq:Ito_SDE}) are not globally Lipschitz continuous, the standard proof of the strong convergence of the EM method is not applicable. We are not aware of any previous rigorous result on the strong convergence of numerical methods for Eq.\,(\ref{eq:Ito_SDE}). The energy-conserving nature of the ESEC method enables us to overcome the difficulty associated with the diverge at the neighborhood of $v=0$ and establish the strong convergence.

The paper is organized as follows. In Sec.~\ref{sec:ESEC_method}, we construct the ESEC method for Eq.\,(\ref{eq:Ito_SDE}) and discuss its properties. In Sec.~\ref{sec:proof_of_convergence}, its strong convergence is rigorously proved. Numerical verification of the convergence rate and energy conservation properties is demonstrated in Sec.\,\ref{sec:numerical_experiments}. Finally, the numerical solution of the pitch-angle scattering in a constant magnetic field is directly compared with the analytical solution in Sec.\,~\ref{sec:constant_B_benchmark}.

\section{\label{sec:ESEC_method} The explicitly solvable energy-conserving algorithm for pitch-angle scattering}

As discussed in Sec.\,\ref{sec:Introduction}, the physics of electrons undergoing pitch-angle scattering in magnetized plasmas can be described by the SDE (\ref{eq:Ito_SDE}). We propose the following Explicitly Solvable and Energy-Conserving (ESEC) one-step method as a numerical algorithm for Eq.\,(\ref{eq:Ito_SDE}), \begin{linenomath} 
\begin{align}
\bar{\boldsymbol{v}}_{k+1}^{\text{ES}}-\boldsymbol{v}_{k} & =\boldsymbol{v}_{k+1/2}\times\boldsymbol{B}(t_{k})h+\sqrt{D_{k}}\left(\dfrac{\boldsymbol{v}_{k}\times\Delta\boldsymbol{W}}{v_{k}^{2}}\right)\times\boldsymbol{v}_{k+1/2},\label{eq:ES_method}\\
\boldsymbol{v}_{k+1/2} & \coloneqq(\boldsymbol{v}_{k}+\bar{\boldsymbol{v}}_{k+1}^{\text{ES}})/2,
\end{align}
\end{linenomath} where $h$ is the step size in time, $t_{k+1}=t_{k}+h$, $D_{k}=D(v_{k})=1/v_{k}$, and $\Delta\boldsymbol{W}=\boldsymbol{W}(t_{k+1})-\boldsymbol{W}(t_{k})\sim\mathcal{N}(0,\boldsymbol{I}h)$ is a 3D Gaussian random variable. It generalizes the previous algorithm for pitch-angle scattering in unmagnetized plasmas \citep{PhysRevE.102.033302}. Taking the dot products of $\boldsymbol{v}_{k+1/2}$ with both sides of Eq.~(\ref{eq:ES_method}), we can see that the ESEC scheme preserves the velocity norm.

The discretization of the Lorentz force is the same as the Boris algorithm \citep{boris1970relativistic,qin2013boris}, and the discretization of the diffusion term can be regarded as a hybrid between the EM method and the midpoint method \citep{milstein2002numerical,milstein2013stochastic}. Although this scheme is implicit, the value of $\bar{\boldsymbol{v}}_{k+1}^{\text{ES}}$ can be explicitly solved. This is because the implicitness is linear, i.e., the right hand side of Eq.\,(\ref{eq:ES_method}) depends on $\bar{\boldsymbol{v}}_{k+1}^{\text{ES}}$ linearly. To solve for $\bar{\boldsymbol{v}}_{k+1}^{\text{ES}},$ we make use of the hat map \begin{linenomath} 
\begin{align}
\boldsymbol{X}=\left(\begin{array}{c}
X_{1}\\
X_{2}\\
X_{3}
\end{array}\right)\longmapsto\hat{\boldsymbol{X}}:=\left(\begin{array}{ccc}
0 & -X_{3} & X_{2}\\
X_{3} & 0 & -X_{1}\\
-X_{2} & X_{1} & 0
\end{array}\right)
\end{align}
\end{linenomath} between a vector and a $3\times3$ skew-symmetric matrix. The cross product between two vectors $\boldsymbol{X}$ and $\boldsymbol{Y}$ can be written as \begin{linenomath} 
\begin{align}
\boldsymbol{X}\times\boldsymbol{Y}=\hat{\boldsymbol{X}}\boldsymbol{Y}=\begin{pmatrix}0 & -X_{3} & X_{2}\\
X_{3} & 0 & -X_{1}\\
-X_{2} & X_{1} & 0
\end{pmatrix}\begin{pmatrix}Y_{1}\\
Y_{2}\\
Y_{3}
\end{pmatrix}.
\end{align}
\end{linenomath} It means that $\hat{\boldsymbol{X}}\boldsymbol{Y}$ is the infinitesimal 3D rotation of $\boldsymbol{Y}$ generated by $\boldsymbol{X}$. The right hand side of Eq.\,(\ref{eq:ES_method}) is thus the infinitesimal 3D rotation of $2\boldsymbol{v}_{k+1/2}$ generated by \begin{linenomath} 
\begin{align}
\boldsymbol{M}_{k}:=\sqrt{D_{k}}\left(\dfrac{\boldsymbol{v}_{k}\times\Delta\boldsymbol{W}}{2v_{k}^{2}}\right)-\dfrac{1}{2}\boldsymbol{B}(t_{k})h.
\end{align}
\end{linenomath} Then the ESEC scheme defined by Eq.~(\ref{eq:ES_method}) can be explicitly solved as \begin{linenomath} 
\begin{align}
\bar{\boldsymbol{v}}_{k+1}^{\text{ES}}=\mathcal{C}(\hat{\boldsymbol{M}}_{k})\boldsymbol{v}_{k},
\end{align}
\end{linenomath} where \begin{linenomath} 
\begin{align}
\mathcal{C}(\hat{\boldsymbol{M}}_{k}):=(\boldsymbol{I}-\hat{\boldsymbol{M}}_{k})^{-1}(\boldsymbol{I}+\hat{\boldsymbol{M}}_{k}),
\end{align}
\end{linenomath} is the Cayley transform of matrix $\hat{\boldsymbol{M}}_{k}$.

It is worth mentioning that the deterministic drift term in the Ito SDE (\ref{eq:Ito_SDE}) does not show up explicitly in the ESEC scheme in Eq.~(\ref{eq:ES_method}). Instead, the drift term is hidden in the implicit part of $\boldsymbol{v}_{k+1/2}$. Since the ESEC scheme is explicitly solvable, if one explicitly calculates the expression of $\boldsymbol{v}_{k+1}$ in terms of $\boldsymbol{v}_{k}$, the deterministic term is recovered. See Appendix B in Ref.\,\citep{PhysRevE.102.033302} for more details of this calculation. More rigorously, in the next section we will calculate the one-step weak error of the ESEC scheme. It shows that the deterministic drift term is recovered from the implicit part with an error in the order of $\mathcal{O}(h^{3/2})$, which guarantees the strong convergence of the ESEC scheme.

As pointed out in Ref.\,\citep{milstein2002numerical}, implicit methods, even when explicitly solvable, for stochastic systems with multiplicative noise may not be acceptable a priori due to its infinite expectation in the one-step approximation. In those situations, a truncation of $\Delta W$ is required \citep{milstein2002numerical}. Such a modification of $\Delta W$ is not necessary for the ESEC scheme. This is because the ESEC scheme preserves the velocity norm, and the expectation of the one-step approximation always exists, i.e., $\mathbb{E}|\bar{\boldsymbol{v}}_{k+1}^{\text{ES}}|=\mathbb{E}|\boldsymbol{v}_{0}|<\infty$.

\section{\label{sec:proof_of_convergence} Proof of strong convergence of the ESEC method}

In this section, we rigorously prove the strong convergence of the ESEC method for Eq.~(\ref{eq:Ito_SDE}) defined on the time interval $t\in[t_{0},T]$ with the initial condition $\boldsymbol{v}=\boldsymbol{v}_{0}$ at $t=t_{0}$. Denote by $\boldsymbol{v}(t;\tilde{t},\tilde{\boldsymbol{v}})$ the analytical solution of the SDE with initial condition $\boldsymbol{v}=\tilde{\boldsymbol{v}}$ at $t=\tilde{t}\in[t_{0},T]$, and by $\bar{\boldsymbol{v}}(\tilde{t}+h;\tilde{t},\tilde{\boldsymbol{v}})$ the one-step approximation from a given value $\boldsymbol{v}=\tilde{\boldsymbol{v}}$ at $t=\tilde{t}\in[t_{0},T-h]$ according to a numerical method, for example, Eq.~(\ref{eq:ES_method}). Using the one-step approximation, we recurrently construct the numerical solution $\bar{\boldsymbol{v}}(t_{k};t_{0},\boldsymbol{v}_{0},h)$ at $t_{k}=t_{0}+kh$ for $k=0,\cdots,N$ with $t_{N}=T$ and $h=(T-t_{0})/N$.

We will prove the following convergence theorem of the ESEC method. 
\begin{thm}
For the SDE (\ref{eq:Ito_SDE}), the numerical solution $\bar{\boldsymbol{v}}(t_{k};t_{0},\boldsymbol{v}_{0},h)$ generated by the ESEC algorithm (\ref{eq:ES_method}) has order $1/2$ strong error in the time interval of $[t_{0},T]$. More precisely, for any $\boldsymbol{v}_{0}\in\mathbb{R}^{3}/\{0\}$, $N\in\mathbb{N}$, and $k=0,1,\cdots,N$, the following inequality holds, \begin{linenomath} 
\begin{align}
\left[\mathbb{E}|\boldsymbol{v}(t_{k};t_{0},\boldsymbol{v}_{0})-\bar{\boldsymbol{v}}(t_{k};t_{0},\boldsymbol{v}_{0},h)|^{2}\right]^{1/2} & \leq K(1+|\boldsymbol{v}_{0}|^{2})^{1/2}h^{1/2},\label{eq:strong_error_whole}
\end{align}
\end{linenomath} where the constant $K$ is independent of $h$. \label{thm} 
\end{thm}
As discussed above, the coefficients in Eq.~(\ref{eq:Ito_SDE}) is not globally Lipschitz continuous and diverge at $\boldsymbol{v}=0$, which makes Theorem \ref{thm} difficult to prove directly. To overcome this obstacle, we consider the following modified SDE for pitch-angle scattering, \begin{linenomath} 
\begin{align}
\mathrm{d}\boldsymbol{v}(t)=\left(\boldsymbol{v}\times\boldsymbol{B}-\alpha^{2}\boldsymbol{v}\right)\mathrm{d}t+\alpha\left(\boldsymbol{I}v-\frac{\boldsymbol{v}\boldsymbol{v}}{v}\right)\cdot\mathrm{d}\boldsymbol{W},\label{eq:Ito_SDE2}
\end{align}
\end{linenomath} where $\alpha$ is a non-zero constant. It can be verified (shown in Appendix \ref{sec:verify_lipschitz}) that the coefficients in Eq.~(\ref{eq:Ito_SDE2}) are all globally Lipschitz continuous and finite at $v=0$. They are also bounded by linear growth as $\boldsymbol{v}\rightarrow\infty$. We further define a modified ESEC method for Eq.~(\ref{eq:Ito_SDE2}) as \begin{linenomath} 
\begin{align}
\bar{\boldsymbol{v}}_{k+1}^{\text{ES}}-\boldsymbol{v}_{k}=\boldsymbol{v}_{k+1/2}\times\boldsymbol{B}(t_{k})h+\alpha\left(\dfrac{\boldsymbol{v}_{k}\times\Delta\boldsymbol{W}}{v_{k}}\right)\times\boldsymbol{v}_{k+1/2}.\label{eq:ES_method2}
\end{align}
\end{linenomath} The modified ESEC method obviously also preserves the velocity norm exactly.

Instead of proving Theorem \ref{thm} directly, we first prove in Lemma \ref{lemma} the convergence of the modified ESEC method (\ref{eq:ES_method2}) for the modified SDE (\ref{eq:Ito_SDE2}), and then show that the lemma implies Theorem \ref{thm} for any non-zero initial condition $\boldsymbol{v}_{0}$. 
\begin{lem}
\label{lemma} For the analytical solution $\boldsymbol{v}(t;t_{0},\boldsymbol{v}_{0})$ of Eq.~(\ref{eq:Ito_SDE2}) and the numerical solution $\bar{\boldsymbol{v}}(t;t_{0},\boldsymbol{v}_{0},h)$ generated by Eq.~(\ref{eq:ES_method2}) in the time interval of $[t_{0},T]$, the following results hold:

(i) The one-step approximation $\bar{\boldsymbol{v}}(t+h;t,\tilde{\boldsymbol{v}},h)$ has order $1$ strong error, i.e., for arbitrary $t\in[t_{0},T-h]$ and $\tilde{\boldsymbol{v}}\in\mathbb{R}^{3}$, \begin{linenomath} 
\begin{align}
\left[\mathbb{E}|\boldsymbol{v}(t+h;t,\tilde{\boldsymbol{v}})-\bar{\boldsymbol{v}}(t+h;t,\tilde{\boldsymbol{v}},h)|^{2}\right]^{1/2} & \leq K(1+|\tilde{\boldsymbol{v}}|^{2})^{1/2}h.\label{eq:strong_error}
\end{align}
\end{linenomath}

(ii) The one-step approximation $\bar{\boldsymbol{v}}(t+h;t,\tilde{\boldsymbol{v}},h)$ has order $3/2$ weak error, i.e., for arbitrary $t\in[t_{0},T-h]$ and $\tilde{\boldsymbol{v}}\in\mathbb{R}^{3}$, \begin{linenomath} 
\begin{align}
\left|\mathbb{E}\left[\boldsymbol{v}(t+h;t,\tilde{\boldsymbol{v}})-\bar{\boldsymbol{v}}(t+h;t,\tilde{\boldsymbol{v}},h)\right]\right| & \leq K(1+|\tilde{\boldsymbol{v}}|^{2})^{1/2}h^{3/2}.\label{eq:weak_error}
\end{align}
\end{linenomath}

(iii) The approximation $\bar{\boldsymbol{v}}(t_{k};t_{0},\boldsymbol{v}_{0})$ has order $1/2$ strong error in the entire time interval, i.e., for any $\boldsymbol{v}_{0}\in\mathbb{R}^{3}$ and $t_{k}=t_{0}+hk$ $(k=1,...,N)$, \begin{linenomath} 
\begin{align}
\left[\mathbb{E}|\boldsymbol{v}(t_{k};t_{0},\boldsymbol{v}_{0})-\bar{\boldsymbol{v}}(t_{k};t_{0},\boldsymbol{v}_{0},h)|^{2}\right]^{1/2} & \leq K(1+|\boldsymbol{v}_{0}|^{2})^{1/2}h^{1/2}.\label{strongErrorWhole}
\end{align}
\end{linenomath} Here, the constant $K$ is independent of $h$ and $\boldsymbol{v}_{0}$. 
\end{lem}
\begin{proof}
The proof starts from the Ito-Taylor expansion (a.k.a. the Wagner-Platen expansion) of Eq.~(\ref{eq:Ito_SDE2}), based on which (i) and (ii) will be proved. We than prove (iii) by applying Theorem 1.1 in Ref.~\citep{milstein2013stochastic}, which states that for a method with order $p_{1}$ strong error and order $p_{2}$ weak error in one step, if $p_{1}\geq1/2$ and $p_{2}\geq p_{1}+1/2$, then the numerical scheme has order $(p_{1}-1/2)$ strong error on the entire time interval.

We rewrite Eq.~(\ref{eq:Ito_SDE2}) using Cartesian indices, \begin{linenomath} 
\begin{align}
\mathrm{d}v_{i} & =\mu_{i}\mathrm{d}t+\sigma_{ir}\mathrm{d}W_{r},\qquad\mu_{i}=\epsilon_{ijk}v_{j}B_{k}-\alpha^{2}v_{i},\quad\sigma_{ir}=\alpha\left(\delta_{ir}v-v_{i}v_{r}/v\right)=\alpha\epsilon_{ijk}\epsilon_{jlr}v_{l}v_{k}/v,\label{eq:Ito_SDE_compact}
\end{align}
\end{linenomath} where $\delta_{ij}$ is the Kronecker delta, $\epsilon_{ijk}$ is the Levi-Civita symbol, and repeated indices are summed over.

When $\boldsymbol{v}(t)$ is a solution of Eq.~(\ref{eq:Ito_SDE_compact}), by Ito formula, any sufficiently smooth function $f(t,\boldsymbol{v})$ can be written as \begin{linenomath} 
\begin{align}
f(t,\boldsymbol{v}(t)) & =f(t_{0},\boldsymbol{v}_{0})+\int_{t_{0}}^{t}\Lambda_{r}f(t_{1},\boldsymbol{v}(t_{1}))\mathrm{d}W_{r}(t_{1})+\int_{t_{0}}^{t}Lf(t_{1},\boldsymbol{v}(t_{1}))\mathrm{d}t_{1},\label{eq:Ito_Taylor}\\
\Lambda_{r}: & =\sigma_{ir}\dfrac{\partial}{\partial x_{i}},\quad L:=\dfrac{\partial}{\partial t}+\mu_{i}\dfrac{\partial}{\partial x_{i}}+\dfrac{1}{2}\sigma_{ir}\sigma_{jr}\dfrac{\partial^{2}}{\partial x_{i}\partial x_{j}}.\label{eq:D}
\end{align}
\end{linenomath} Consecutively applying Eq.~(\ref{eq:Ito_Taylor}) to the coefficients inside the integral leads to the Ito-Taylor expansion with integral type remainders.

The Ito-Taylor expansion for $v_{i}(t)$ itself is \begin{linenomath} 
\begin{align}
\begin{split}v_{i}(t+h)=v_{i}(t)+\alpha & \epsilon_{ijk}v_{j}(t)B_{k}(t)h-\alpha^{2}v_{i}(t)h+\alpha\epsilon_{ijk}\epsilon_{jlr}v_{l}(t)v_{k}(t)\Delta W_{r}/v(t)+\sum_{\lambda=1}^{4}\rho_{i,\lambda},\end{split}
\label{eq:exact}
\end{align}
\end{linenomath} where $\Delta W_{r}=\int_{t}^{t+h}\mathrm{d}W_{r}\sim\mathcal{N}(0,h)$. The first three terms of the expansion is the Euler-Maruyama approximation of Eq.~(\ref{eq:Ito_SDE2}); the forth term is the four integral remainders, \begin{linenomath} 
\begin{align}
\begin{split}\rho_{i,1} & =\int_{t}^{t+h}\int_{t}^{t_{1}}L\mu_{i}(t_{2})\mathrm{d}t_{2}\mathrm{d}t_{1},\quad\rho_{i,2}=\int_{t}^{t+h}\int_{t}^{t_{1}}\Lambda_{m}\mu_{i}(t_{2})\mathrm{d}W_{m}(t_{2})\mathrm{d}t_{1},\\
\rho_{i,3} & =\int_{t}^{t+h}\int_{t}^{t_{1}}L\sigma_{ir}(t_{2})\mathrm{d}t_{2}\mathrm{d}W_{r}(t_{1}),\quad\rho_{i,4}=\int_{t}^{t+h}\int_{t}^{t_{1}}\Lambda_{m}\sigma_{ir}(t_{2})\mathrm{d}W_{m}(t_{2})\mathrm{d}W_{r}(t_{1}),
\end{split}
\end{align}
\end{linenomath} where $\mu_{i}(t)\equiv\mu_{i}(\boldsymbol{v}(t),t)$ and $\sigma_{ir}(t)\equiv\sigma_{ir}(\boldsymbol{v}(t),t)$.

Since $\mu_{i}$ and $\sigma_{ir}$ are well-behaved, it can be verified that functions $L\mu_{i}$, $\Lambda_{m}\mu_{i}$, $L\sigma_{ir}$, and $\Lambda_{m}\sigma_{ir}$ all have linear growth bound and are finite when $v\to0$. Thus, the estimation of the magnitude of all the remainders can be immediately found using Lemma 2.2 in Ref. \citep{milstein2013stochastic}.

In the index notation, we denote by $|v_{i}|=(v_{x}^{2}+v_{y}^{2}+v_{z}^{2})^{1/2}$ the Euclidean norm of the vector $\boldsymbol{v}$. The following mean-squared estimation holds, \begin{linenomath} 
\begin{align}
\mathbb{E}|\rho_{i,\lambda}|^{2}\leq K(1+|v_{i}(t)|^{2})h^{p_{\lambda}},\label{eq:integral_estimation}
\end{align}
\end{linenomath} where $p_{1}=4$, $p_{2}=p_{3}=3$, $p_{4}=2$, and $K$ is a constant independent of $h$ and $v_{i}(t)$. In addition, since $\rho_{i,2}$, $\rho_{i,3}$, $\rho_{i,4}$ are Ito integrals, we have \begin{linenomath} 
\begin{align}
\mathbb{E}\rho_{i,2}=\mathbb{E}\rho_{i,3}=\mathbb{E}\rho_{i,4}=0.
\end{align}
\end{linenomath}

(i) One-step strong error.

Denote the one-step approximation by $\bar{\boldsymbol{v}}(t+h)\equiv\bar{\boldsymbol{v}}(t+h;t,\boldsymbol{v}(t))$, and let $\Delta W_{r}=\xi_{r}\sqrt{h}$ where $\xi_{r}\sim\mathcal{N}(0,1)$. The modified ESEC method Eq.~(\ref{eq:ES_method2}) can be written using the Cartesian indices as \begin{linenomath} 
\begin{align}
\bar{v}_{i}(t+h)-v_{i}(t)= & \epsilon_{ijk}\dfrac{v_{j}(t)+\bar{v}_{j}(t+h)}{2}B_{k}(t)h+\alpha\epsilon_{ijk}\epsilon_{jlr}v_{l}(t)\sqrt{h}\xi_{r}\dfrac{v_{k}(t)+\bar{v}_{k}(t+h)}{2v(t)},\nonumber \\
= & \epsilon_{ijk}M_{j}\sqrt{h}\,[v_{k}(t)+\bar{v}_{k}(t+h)],\label{eq:ES_method_compact}\\
M_{j}:= & \alpha\dfrac{\epsilon_{jlr}v_{l}(t)\xi_{r}}{2v(t)}-\dfrac{1}{2}B_{j}(t)\sqrt{h}.\label{eq:M_j}
\end{align}
\end{linenomath}

The difference between the one-step approximation (\ref{eq:ES_method_compact}) and the exact solution (\ref{eq:exact}) is \begin{linenomath} 
\begin{align}
\bar{v}_{i}(t+h)-v_{i}(t+h) & =\epsilon_{ijk}M_{j}\sqrt{h}\big[\bar{v}_{k}(t+h)-v_{k}(t)\big]+\alpha^{2}v_{i}(t)h-\sum_{\lambda}\rho_{i,\lambda},\nonumber \\
 & =\epsilon_{ijk}\epsilon_{kpq}M_{j}M_{p}{h}\,[v_{q}(t)+\bar{v}_{q}(t+h)]+\alpha^{2}v_{i}(t)h-\sum_{\lambda}\rho_{i,\lambda},\label{eq:one_step_error_1}
\end{align}
\end{linenomath} where in the second line, we have used Eq.~(\ref{eq:ES_method_compact}) again to replace the term inside the square bracket. The one-step strong error is given by $\mathbb{E}|\bar{v}_{i}(t+h)-v_{i}(t+h)|^{2}$. Using the inequality$|X_{i}+Y_{i}|^{2}\leq2|X_{i}|^{2}+2|Y_{i}|^{2}$, we can estimate the mean-square bound of the RHS of Eq.~(\ref{eq:one_step_error_1}) term by term. Because $|\epsilon_{ijk}X_{j}Y_{k}|\leq|X_{j}|\cdot|Y_{k}|$, the mean square of the first term is \begin{linenomath} 
\begin{align}
R_{s,1} & :=\mathbb{E}\big|\epsilon_{ijk}\epsilon_{kpq}M_{j}M_{p}{h}\,[v_{q}(t)+\bar{v}_{q}(t+h)]\big|^{2}\leq h^{2}\mathbb{E}\big[|M_{j}|^{4}|v_{q}(t)+\bar{v}_{q}(t+h)|^{2}\big]\\
 & \leq4h^{2}|v_{q}(t)|^{2}\mathbb{E}\big[|M_{j}|^{4}\big],
\end{align}
\end{linenomath} where use is made of $|v_{q}(t)+\bar{v}_{q}(t+h)|\leq|v_{q}(t)|+|\bar{v}_{q}(t+h)|=2|v_{q}(t)|$ in the second line. The bound for $|M_{j}|$ can be calculated as follows. \begin{linenomath} 
\begin{align}
\left|\alpha\dfrac{\epsilon_{jlr}v_{l}(t)\xi_{r}}{2v(t)}-\dfrac{1}{2}B_{j}(t)\sqrt{h}\right|\leq\dfrac{\alpha}{2}\left|\dfrac{\epsilon_{jlr}v_{l}(t)\xi_{r}}{v(t)}\right|+\dfrac{\sqrt{h}}{2}|B_{j}(t)|\leq\dfrac{1}{2}(\alpha|\xi_{r}|+\sqrt{h}|B_{j}(t)|),
\end{align}
\end{linenomath} where we have used $|v_{l}(t)/v(t)|=1$. Therefore, \begin{linenomath} 
\begin{align}
R_{s,1} & \leq\dfrac{h^{2}}{4}|v_{q}(t)|^{2}\mathbb{E}\left[(\alpha|\xi_{r}|+\sqrt{h}|B_{j}(t)|)^{4}\right]\nonumber \\
 & =\dfrac{h^{2}}{4}|v_{q}(t)|^{2}\left[\alpha^{4}\mathbb{E}|\xi_{r}|^{4}+\mathcal{O}(\sqrt{h})\right]\leq K\,|v_{q}(t)|^{2}h^{2}.
\end{align}
\end{linenomath} Here, the boundedness of $B(t)$ has been used and $K$ is a constant independent of $h$ and $v_{q}(t)$.

The estimation of the mean square of the second term in Eq.~(\ref{eq:one_step_error_1}) is straightforward; the estimation of that of the third term is given by Eq.~(\ref{eq:integral_estimation}). Putting the estimations of all three terms together proves Eq.~(\ref{eq:strong_error}), i.e., the one-step strong error of the modified ESEC scheme is at least order $1$.

(ii) One-step weak error.

Next, we calculate the one-step weak error $|\mathbb{E}[\bar{v}_{i}(t+h)-v_{i}(t+h)]|$. Plugging the expression of $\bar{v}_{i}(t+h)$ from Eq.~(\ref{eq:ES_method_compact}) into the RHS of Eq.~(\ref{eq:one_step_error_1}) again, we get \begin{linenomath} 
\begin{align}
\begin{split}\bar{v}_{i}(t+h)-v_{i}(t+h)= & \,\quad2\epsilon_{ijk}\epsilon_{kpq}M_{j}M_{p}hv_{q}(t)+\alpha^{2}v_{i}(t)h\\
 & +\epsilon_{ijk}\epsilon_{kpq}\epsilon_{qmn}M_{j}M_{p}M_{m}h^{3/2}[v_{n}(t)+\bar{v}_{n}(t+h)]\\
 & -\sum_{\lambda}\rho_{i,\lambda}.
\end{split}
\label{eq:one_step_error_2}
\end{align}
\end{linenomath}

With the intention to apply the inequality $|\mathbb{E}(X+Y)|\leq|\mathbb{E}(X)|+|\mathbb{E}(Y)|$, we estimate the mean of the RHS of Eq.~(\ref{eq:one_step_error_2}) line by line. Since $|\mathbb{E}X|\leq\mathbb{E}|X|$, the bound of the mean for the second line is \begin{linenomath} 
\begin{align}
R_{w,2} & :=\left|\mathbb{E}\big[\epsilon_{ijk}\epsilon_{kpq}\epsilon_{qmn}M_{j}M_{p}M_{m}h^{3/2}[v_{n}(t)+\bar{v}_{n}(t+h)]\big]\right|\nonumber \\
 & \leq\mathbb{E}\big|\epsilon_{ijk}\epsilon_{kpq}\epsilon_{qmn}M_{j}M_{p}M_{m}h^{3/2}[v_{n}(t)+\bar{v}_{n}(t+h)]\big|\nonumber \\
 & \leq h^{3/2}\mathbb{E}\big[|M_{j}|^{3}|v_{n}(t)+\bar{v}_{n}(t+h)|\big]\leq\dfrac{h^{3/2}}{4}|v_{n}(t)|\mathbb{E}\left[(\alpha|\xi_{r}|+\sqrt{h}|B_{j}(t)|)^{3}\right]\nonumber \\
 & =\dfrac{h^{3/2}}{4}|v_{n}(t)|\left[\alpha^{3}\mathbb{E}|\xi_{r}|^{3}+\mathcal{O}(\sqrt{h})\right]\leq K|v_{n}(t)|h^{3/2}.
\end{align}
\end{linenomath}

In the third line, $\rho_{i,2}$, $\rho_{i,3}$, $\rho_{i,4}$ vanish when taking expectation. The bound of the mean for $\rho_{i,1}$ can be estimated by the Cauchy-Schwarz inequality, \begin{linenomath} 
\begin{align}
|\mathbb{E}\rho_{i,1}|\leq\mathbb{E}|\rho_{i,1}|\leq\sqrt{\mathbb{E}|\rho_{i,1}|^{2}}\leq K(1+|v_{i}(t)|^{2})^{1/2}h^{2},
\end{align}
\end{linenomath}

For the first line, plugging in the definition of $M_{j}$ and after some algebra, we have \begin{linenomath} 
\begin{align}
R_{w,1}:= & \mathbb{E}L_{1}\nonumber \\
L_{1}:= & \,2\epsilon_{ijk}\epsilon_{kpq}M_{j}M_{p}hv_{q}(t)+\alpha^{2}v_{i}(t)h\nonumber \\
\begin{split}= & \,2\alpha^{2}\epsilon_{ijk}\epsilon_{kpq}\dfrac{\epsilon_{jlr}v_{l}(t)\xi_{r}}{2v(t)}\dfrac{\epsilon_{pmn}v_{m}(t)\xi_{m}}{2v(t)}v_{q}(t)h+\alpha^{2}v_{i}(t)h\\
 & -\alpha\dfrac{h^{3/2}}{2v(t)}\epsilon_{ijk}\epsilon_{kpq}\Big[\epsilon_{jlr}v_{l}(t)\xi_{r}B_{p}(t)+\epsilon_{pmn}v_{m}(t)\xi_{n}B_{j}(t)\Big]v_{q}(t)\\
 & +\dfrac{h^{2}}{2}\epsilon_{ijk}\epsilon_{kpq}B_{j}(t)B_{p}(t)v_{q}(t).
\end{split}
\label{eq:one_step_error_3}
\end{align}
\end{linenomath} In Eq.~(\ref{eq:one_step_error_3}), the third line vanishes under expectation because $\mathbb{E}\xi_{r}=0$; the fourth line is on the order of $\mathcal{O}(h^{2})$. To calculate the mean of the first term in the second line, notice that $\mathbb{E}(\xi_{i}\xi_{j})=\delta_{ij}$, and we have \begin{linenomath} 
\begin{align}
 & \mathbb{E}\left[2\alpha^{2}(\epsilon_{ijk}\epsilon_{jlr})(\epsilon_{kpq}\epsilon_{pmn})\dfrac{v_{l}(t)\xi_{r}}{2v(t)}\dfrac{v_{m}(t)\xi_{n}}{2v(t)}v_{q}(t)h\right]\\
= & \alpha^{2}(\delta_{kl}\delta_{ir}-\delta_{kr}\delta_{il})(\delta_{qm}\delta_{kn}-\delta_{qn}\delta_{km})\left[\dfrac{v_{l}(t)v_{m}(t)v_{q}(t)}{2v^{2}(t)}\right]\mathbb{E}(\xi_{r}\xi_{n})h\\
= & -\alpha^{2}v_{i}(t)h,
\end{align}
\end{linenomath} where use is made of $\delta_{ii}=3$. Notice that the expectation, $-\alpha^{2}v_{i}(t)h$, is exactly the deterministic drift in the Ito-Taylor expansion in Eq.~(\ref{eq:exact}) and cancels out the second term in the second line of Eq.~(\ref{eq:one_step_error_3}). Hence, the second line in Eq.\,(\ref{eq:one_step_error_3}) vanishes under expectation, and we have 
\begin{equation}
R_{w,1}=\mathcal{O}(h^{2}).
\end{equation}

Combining all the estimations for terms on the RHS of Eq.~(\ref{eq:one_step_error_2}) proves Eq.~(\ref{eq:weak_error}), i.e., the one-step weak error of the modified ESEC method is at least order $3/2$.

(iii) Strong convergence of the ESEC method.

We have proved that the strong and weak errors in one step between the modified ESEC method in Eq.~(\ref{eq:ES_method2}) and the solution of Ito SDE Eq.~(\ref{eq:Ito_SDE2}) are on order $1$ and order $3/2$, respectively. In addition, the coefficients of Eq.~(\ref{eq:Ito_SDE2}) are globally Lipschitz continuous (see Appendix \ref{sec:appenix}) and have linear growth bound. Based on Theorem 1.1 in Ref. \citep{milstein2013stochastic}, the numerical solutions constructed by Eq.~(\ref{eq:ES_method2}) converge to exact solutions of Eq.~(\ref{eq:Ito_SDE2}) with order $1/2$ strong error on the entire time interval.

This completes the proof of Lemma \ref{lemma}. 
\end{proof}
We now prove Theorem \ref{thm}. 
\begin{proof}
Notice that for both Eq.~(\ref{eq:Ito_SDE}) and Eq.~(\ref{eq:Ito_SDE2}), the velocity norm is preserved by the solution, which can be verified by applying the Ito formula \citep{kloeden2013numerical} to $v$, \begin{linenomath} 
\begin{align}
\mathrm{d}v=\left[\dfrac{\partial v}{\partial t}+\boldsymbol{\mu}\cdot\dfrac{\partial v}{\partial\boldsymbol{v}}+\dfrac{1}{2}\Tr\left(\boldsymbol{\sigma}^{\mathrm{T}}\dfrac{\partial^{2}v}{\partial\boldsymbol{v}\partial\boldsymbol{v}}\boldsymbol{\sigma}\right)\right]\mathrm{d}t+\dfrac{\partial v}{\partial\boldsymbol{v}}\cdot\boldsymbol{\sigma}\cdot\mathrm{d}\boldsymbol{W}=0,\label{eq:dv=00003D00003D0}
\end{align}
\end{linenomath} where $\boldsymbol{\mu}$ and $\boldsymbol{\sigma}$ are the drift and diffusion coefficients in either Eq.~(\ref{eq:Ito_SDE}) or Eq.~(\ref{eq:Ito_SDE2}). Thus any analytical solution $\boldsymbol{v}(t)$ of Eq.~(\ref{eq:Ito_SDE}) with initial condition $\boldsymbol{v}(t_{0})=\boldsymbol{v}_{0}\neq0$ is also the solution of Eq.~(\ref{eq:Ito_SDE2}) with the same initial condition if we choose the constant $\alpha=\sqrt{D(v_{0})}/v_{0}$.

On the other hand, both the ESEC method given by Eq.~(\ref{eq:ES_method}) for Eq.\,(\ref{eq:Ito_SDE}) and the modified ESEC method by Eq.~(\ref{eq:ES_method2}) for Eq.~(\ref{eq:Ito_SDE2}) preserve the velocity norm exactly. For any given initial condition $\boldsymbol{v}_{0}\neq0$, the numerical solution generated by the ESEC method (\ref{eq:ES_method}) for Eq.~(\ref{eq:Ito_SDE}) is exactly the same as that by the modified ESEC method (\ref{eq:ES_method2}) for Eq.~(\ref{eq:Ito_SDE2}), if we choose $\alpha=\sqrt{D(v_{0})}/v_{0}$.

According to Lemma \ref{lemma}, the modified ESEC method (\ref{eq:ES_method2}) generates numerical solutions converging to the exact solutions of Eq.~(\ref{eq:Ito_SDE2}) with order $1/2$ strong error. Combining these facts, we conclude that the ESEC method (\ref{eq:ES_method}) generates numerical solutions converging to the exact solutions of Eq.~(\ref{eq:Ito_SDE}) with order $1/2$ strong error.

This completes the proof of Theorem \ref{thm} 
\end{proof}
Two comments are in order. First, in Lemma \ref{lemma}, the constant $K$ is independent of initial condition $\boldsymbol{v}_{0}$, but depends on the constant $\alpha$. In order to make Eq.~(\ref{eq:Ito_SDE}) and Eq.~(\ref{eq:Ito_SDE2}) equivalent, the choice of $\alpha$ depends on the initial condition, which makes the constant $K$ in Theorem \ref{thm} dependent on initial conditions. Second, the proof relies on the fact that the velocity norm $v$ is conserved for both analytical solutions and numerical solutions. For the standard EM method implemented in Cartesian coordinates, the numerical solution does not preserve the velocity norm, and as a consequence the convergence cannot be established. In Sec.\,\ref{sec:divergent_sample_path}, divergent sample paths of the EM method will be demonstrated.

\section{\label{sec:numerical_experiments} Numerical verfication of convergence }

\subsection{Strong and weak convergence}

In this and the following sections, we compare the ESEC scheme and the EM scheme implemented in Cartesian coordinates (labeled by EM\_C). Since the analytical solution of Eq.~(\ref{eq:Ito_SDE}) is unknown, it is not feasible to calculate the strong and weak errors according to Eqs.~(\ref{eq:strong_error}) and (\ref{eq:weak_error}). Thus we define the following relative errors between different step sizes $h_{l}$ and $h_{l+1}$, \begin{linenomath} 
\begin{align}
\bar{\epsilon}_{\text{s}}(t;t_{0},\boldsymbol{v}_{0},h_{l}) & :=\left[\mathbb{E}|\bar{\boldsymbol{v}}(t;t_{0},\boldsymbol{v}_{0},h_{l+1})-\bar{\boldsymbol{v}}(t;t_{0},\boldsymbol{v}_{0},h_{l})|^{2}\right]^{1/2},\label{eq:strong_error2}\\
\bar{\epsilon}_{\text{w}}(t;t_{0},\boldsymbol{v}_{0},h_{l}) & :=\left|\mathbb{E}\left[\bar{\boldsymbol{v}}(t;t_{0},\boldsymbol{v}_{0},h_{l+1})-\bar{\boldsymbol{v}}(t;t_{0},\boldsymbol{v}_{0},h_{l})\right]\right|.\label{eq:weak_error2}
\end{align}
\end{linenomath} It is easy to prove that for a numerical scheme with order $p_{\text{w}}$ weak error and order $p_{\text{s}}$ strong error, $\bar{\epsilon}_{\text{s/w}}\sim\mathcal{O}(h_{l}^{p_{\text{s/w}}})$ as long as $h_{l+1}<h_{l}$ (see Appendix \ref{sec:appenix}).

\begin{figure}[ht]
\centering \begin{subfigure}[b]{0.4\textwidth} \centering \includegraphics[width=1\textwidth]{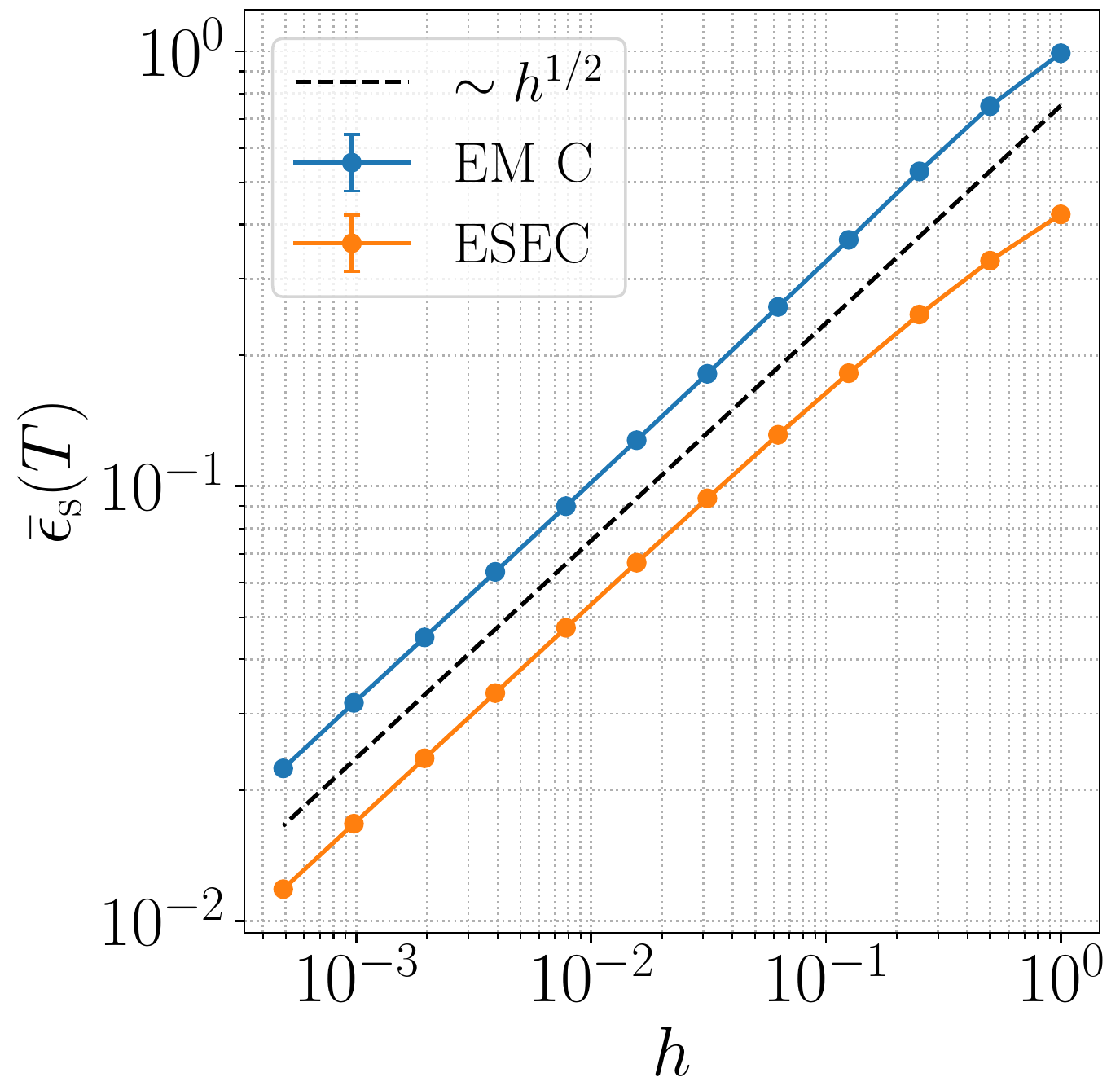} \caption{Strong errors}
\label{fig:strong_error} \end{subfigure} \hspace{0.05\textwidth} \begin{subfigure}[b]{0.4\textwidth} \centering \includegraphics[width=1\textwidth]{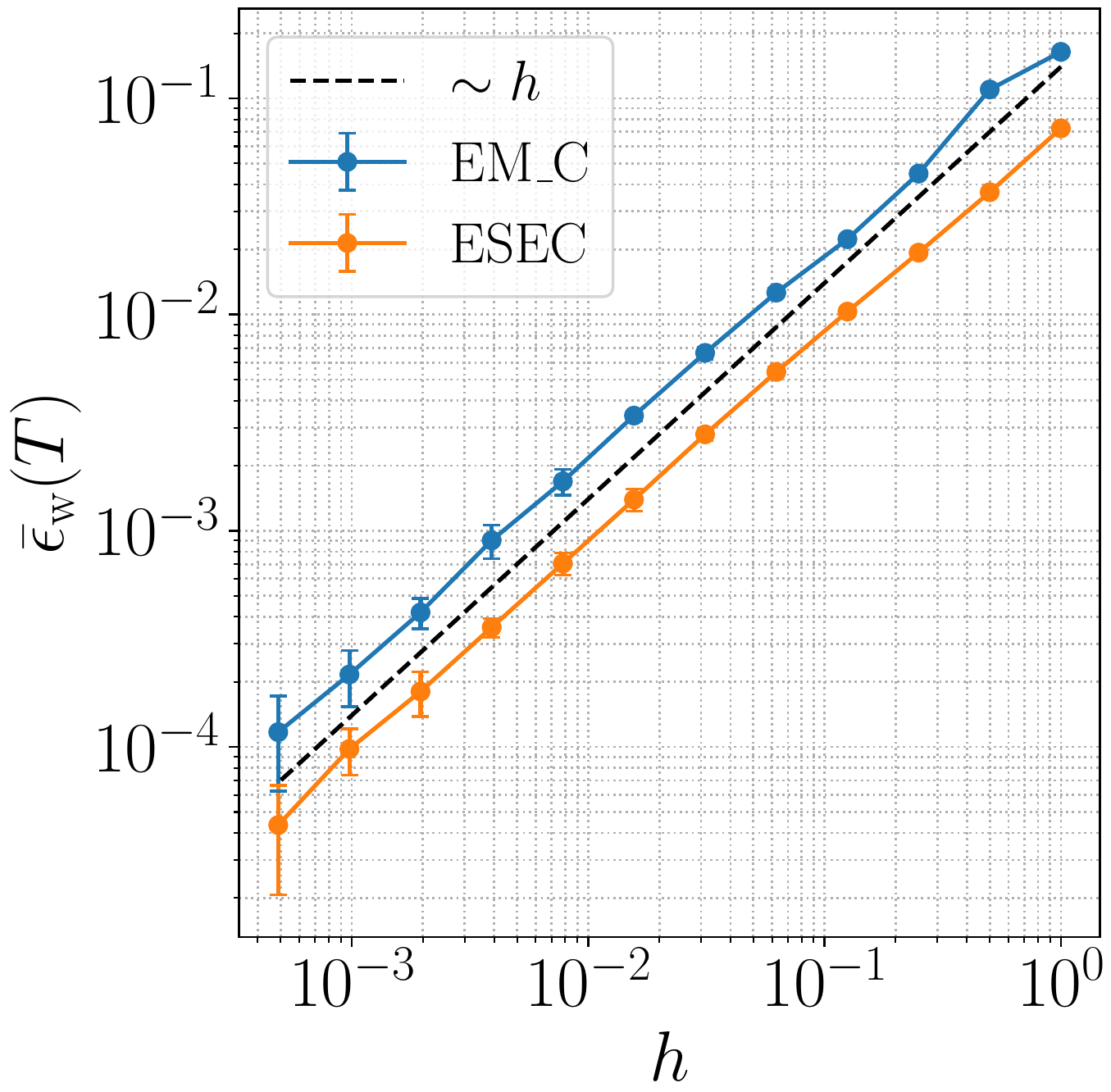} \caption{Weak errors}
\label{fig:weak_error} \end{subfigure} \caption{Strong (a) and weak (b) errors for EM\_C and ESEC schemes. Reference lines with order $1/2$ and $1$ are shown as dashed lines. The initial condition is $\boldsymbol{v}_{0}=(0,0,1)$ and the magnetic field is $\boldsymbol{B}=(0,0,1)$. The error bars in strong errors are too small to be visible. The sample size is $5\times10^{6}$.}
\label{fig:strong_and_weak_errors}
\end{figure}

With the definition and fact above, the relative strong and weak errors at the end time $t=T=1$ for $h_{l}=2^{-l}$ $(l=1\cdots12)$ are plotted in Fig.~\ref{fig:strong_and_weak_errors}. Figure \ref{fig:strong_error} shows that both the ESEC and EM\_C methods have order $1/2$ strong error, but the strong convergence for the EM\_C method cannot be rigorously established due to the violation of the global Lipschitz condition. Although we did not establish the weak order of the two schemes for Eq.\,(\ref{eq:Ito_SDE}) in Sec.\,\ref{sec:proof_of_convergence}, Fig.\,\ref{fig:weak_error} shows that the both schemes have order $1$ weak error in the entire time interval. In addition, the ESEC scheme has smaller scale constants for both strong and weak errors than the EM\_C scheme.

\subsection{\label{sec:conserve_norm} Conservation of the velocity norm}

To compare the conservation of the velocity norm $v$, we define the following strong error of velocity magnitude, \begin{linenomath} 
\begin{align}
\epsilon_{v}(t;t_{0},\boldsymbol{v}_{0},h):=\big[\mathbb{E}|\bar{v}(t;t_{0},\boldsymbol{v}_{0},h)-v_{0}|^{2}\big]^{1/2},
\end{align}
\end{linenomath} where $\bar{v}(t;t_{0},\boldsymbol{v}_{0},h)$ is the numerically calculated velocity norm at time $t$ with initial condition $\boldsymbol{v}_{0}$ at $t=t_{0}$ and step size $h$.

\begin{figure}[ht]
\centering \begin{subfigure}[t]{0.4\textwidth} \centering \includegraphics[height=6.5cm]{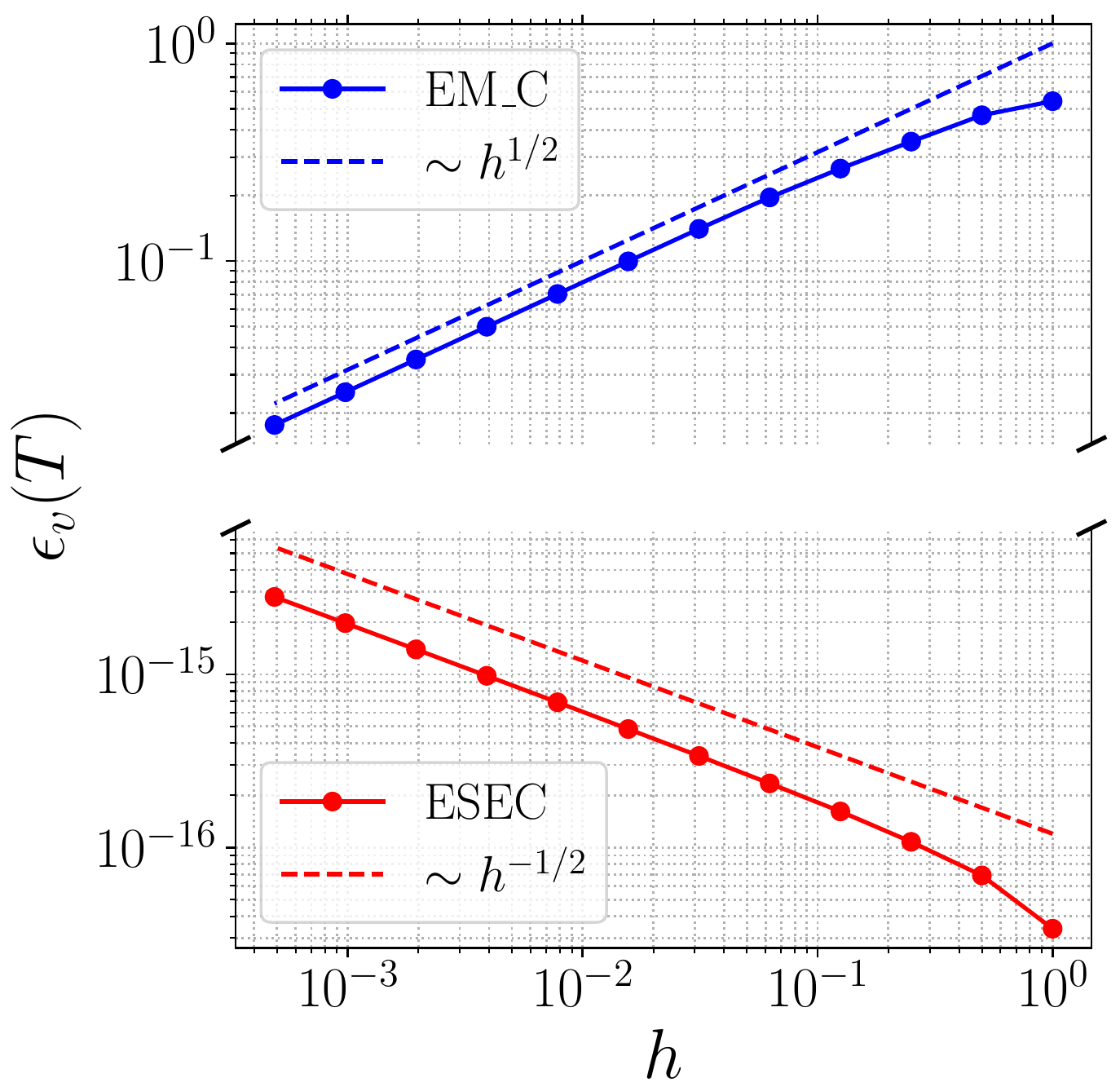} \caption{Strong errors of $v$ as a function of the time-step size $h$ at the fixed end time $t$.}
\label{fig:strong_error_v} \end{subfigure} \hspace{0.05\textwidth} \begin{subfigure}[t]{0.4\textwidth} \centering \includegraphics[height=6.5cm]{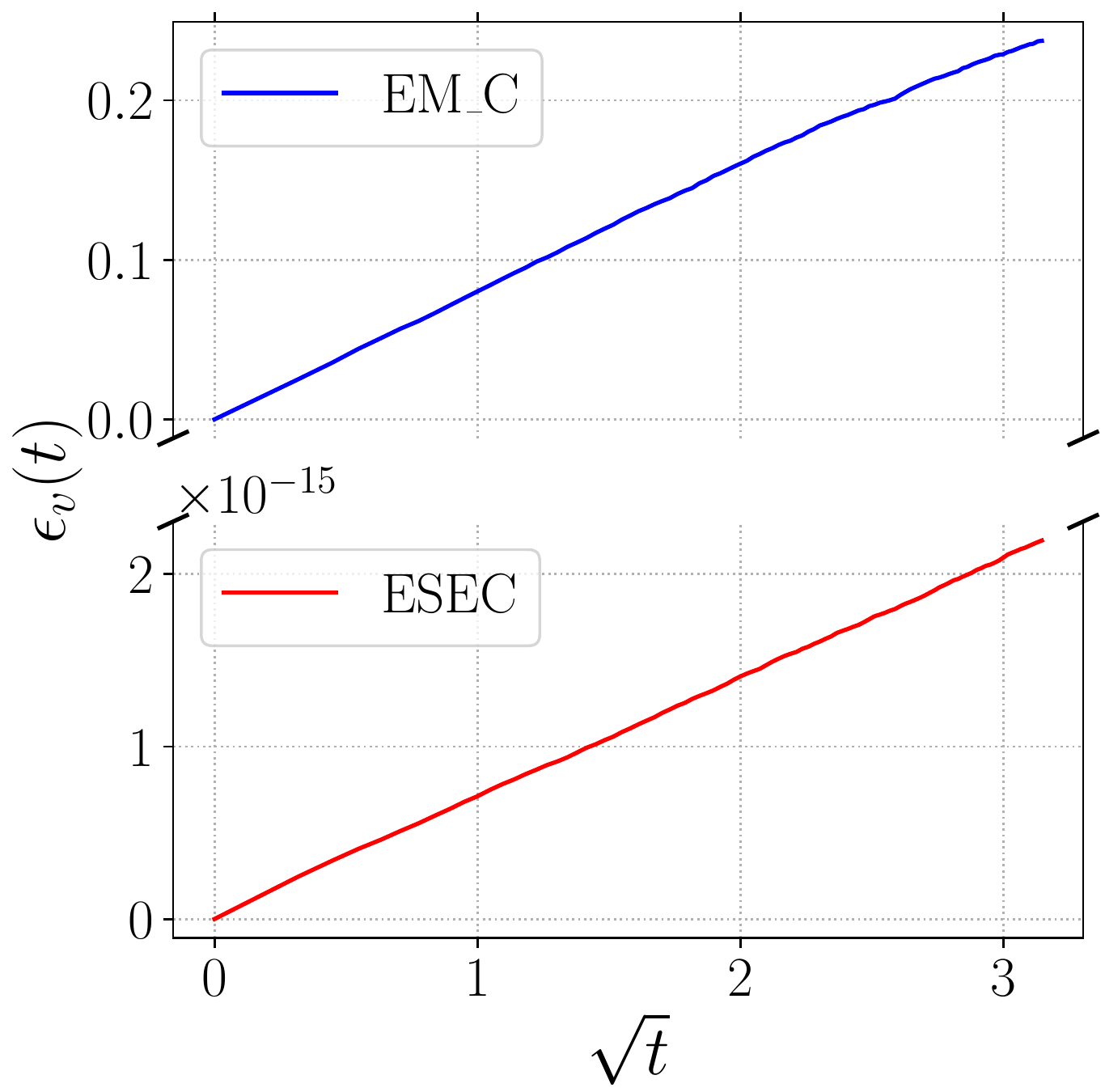} \caption{Strong errors of $v$ as a function of $\sqrt{t}$ with a fixed time-step size $h$.}
\label{fig:strong_error_v2} \end{subfigure} \caption{The strong errors of velocity norm. The vertical axes are separated into two parts because the error of the two schemes are different by more than 10 orders of magnitude. The initial condition is $\boldsymbol{v}_{0}=(0,0,1)$ and the magnetic field is $\boldsymbol{B}=(0,0,1)$.}
\end{figure}

Figure \ref{fig:strong_error_v} shows the strong errors of the velocity norm $\epsilon_{v}$ at the fixed end time $t=T=3$ as a function of $h\in[2^{-12},2^{-1}]$. For the EM\_C method, $\epsilon_{v}\sim h^{1/2}$, consistent with its usual strong error. For the ESEC method, due to its energy-preserving property, the magnitude of $\epsilon_{v}$ is the same as the machine precision which is $10^{-15}$. The interesting result that $\epsilon_{v}\sim h^{-1/2}$ can be understood as follows. Because of machine error, the value of $\delta v=v-v_{0}$ performs a random walk at each time step. For a calculation with total number time steps $L=T/h$, we have $|\delta v(T)|\sim L^{1/2}\sim h^{-1/2}$. Shown in Fig.~\ref{fig:strong_error_v2} are the strong errors of the velocity norm with a fixed time step $h=0.01$ for $t\in[0,10]$. We can see that for both schemes $\epsilon_{v}\sim h^{1/2}$. This is because the errors at different time steps are independent and the accumulated errors behave like random walks. The difference between the two methods is that the error for the EM\_C method is its truncation error, but that for the ESEC method is the machine error.

In Fig.~\ref{fig:strong_error_v2}, we observe that at $t=10$, the strong error of the velocity norm for the EM\_C method is more than 20\% of its initial value. This type of error in the velocity norm is deleterious. For some sample paths, it can push $v$ towards zero, where the drift and diffusion coefficients are singular, which in turn leads to the divergence of the sample paths. Such sample paths will be numerically demonstrated in the next section. To avoid this divergence, one method is to regularize the coefficients of the SDE. For example, one can modify the coefficients when $v$ is smaller than a chosen critical value $v_{c}$ \citep{rosin2014multilevel} (see Appendix \ref{sec:rEM_C} for more details). 

Another choice is to switch to spherical coordinates. Since the velocity magnitude is a constant of motion according to Eq.\,(\ref{eq:Ito_psi}), it conserved exactly without machine error. In this sense, the EM method implemented in spherical coordinates has an even better energy-conservation property than the ESEC method, if one is not concerned with the singularity at $\mu=\pm1$ for Eq.\,(\ref{eq:Ito_psi}). In practice, regularization techniques in spherical coordinates are still necessary for the dynamics of azimuthal angle described by Eq.\,(\ref{eq:Ito_psi}) (see Appendix \ref{sec:rEM_s} for more details). In comparison, modifications of the governing equations for the purpose of avoiding numerical singularity are not necessary for the ESEC method.

\section{\label{sec:constant_B_benchmark} pitch-angle scattering in a constant magnetic field}

In this section we benchmark our algorithm against an analytical solvable problem. Consider the electron-ion collisions of a magnetized plasma in a constant magnetic field. The electrons have a uniform initial velocity $\boldsymbol{v}_{0}$ and the ions are cold and stationary. In the limit of $m_{e}/m_{i}\to0$, the electron trajectory can be described by Ito SDE (\ref{eq:Ito_SDE}), and the time evolution of the electron distribution function $f_{e}$ is described by the FP equation (\ref{FP}). For comparison, the problem will be calculated using the ESEC algorithm, the normal EM algorithm in Cartesian coordinates (refer to `EM\_C'), the regularized EM algorithm in Cartesian coordinates (refer to `rEM\_C', see Appendix \ref{sec:rEM_C}), and the regularized EM algorithm in spherical coordinates (refer to `rEM\_s', see Appendix \ref{sec:rEM_s}).

\subsection{\label{sec:divergent_sample_path} Single-electron trajectories }

Firstly, we compare the electron trajectory calculated by different algorithms with the same underlying Wiener process. The method to transform the same Wiener process to different coordinates is discussed in Appendix \ref{sec:generate_weiner}. Sample paths of $v$ and $\mu$ calculated with different algorithms are shown in Fig.~\ref{fig:exploded}.

\begin{figure}[ht]
\centering \begin{subfigure}[t]{0.45\textwidth} \centering \includegraphics[width=1\textwidth]{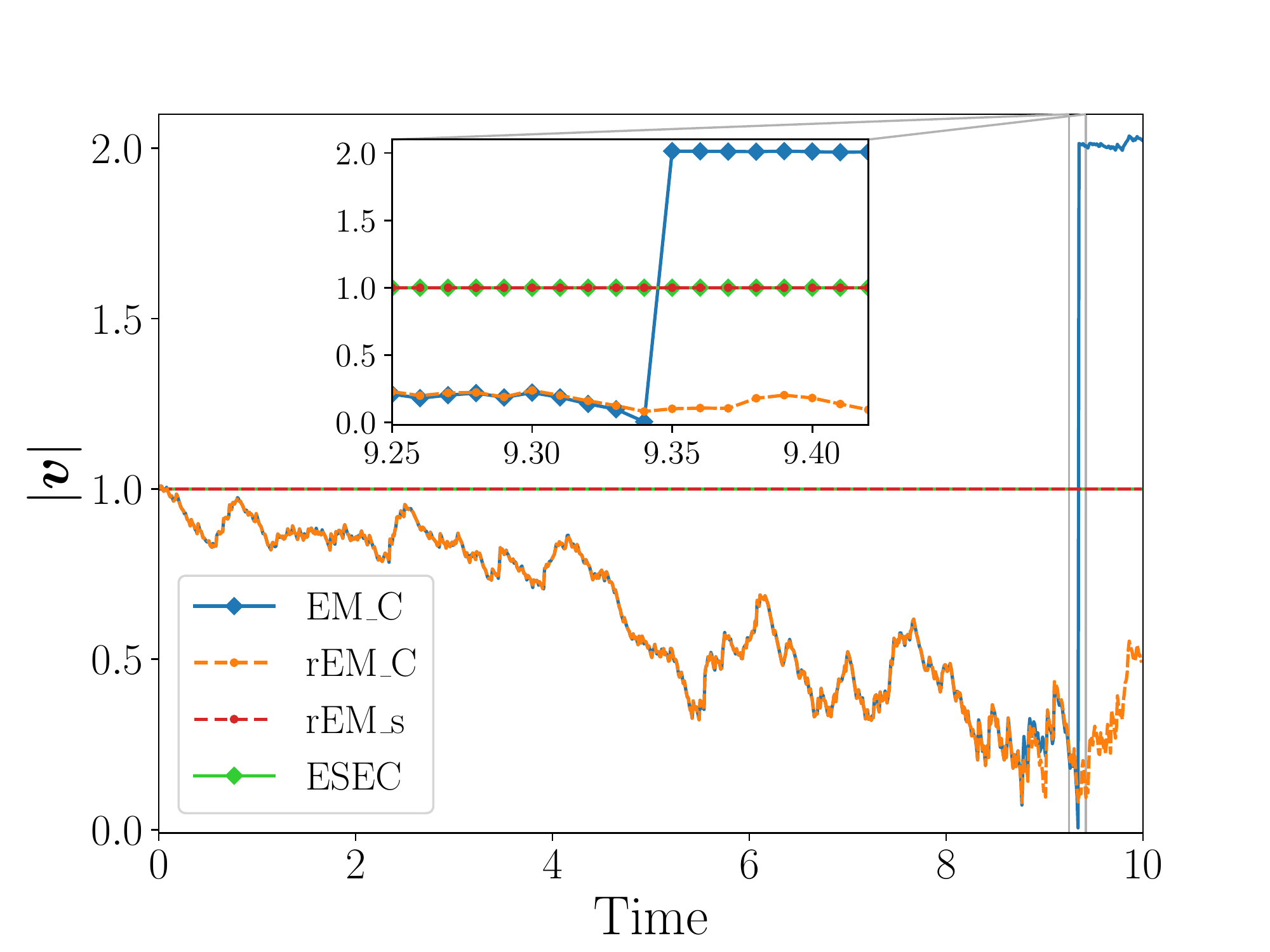} \caption{Sample paths of $v(t)$.}
\label{fig:exploded_v} \end{subfigure} \hspace{0.05\textwidth} \begin{subfigure}[t]{0.45\textwidth} \centering \includegraphics[width=1\textwidth]{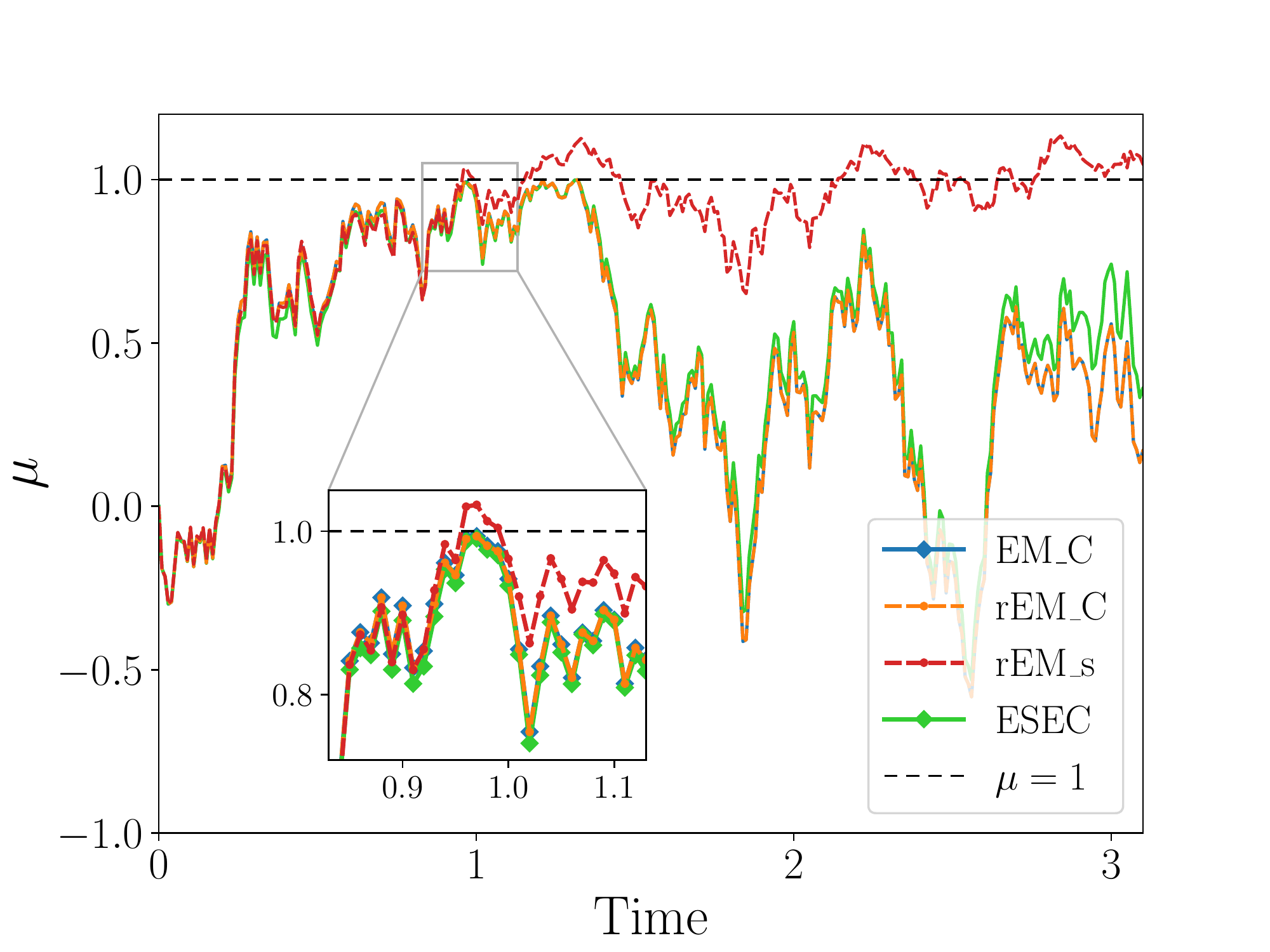} \caption{Sample paths of $\mu(t)$.}
\label{fig:exploded_mu} \end{subfigure} \caption{Sample paths calculated with the standard EM in Cartesian coordinates (blue), the regularized EM in Cartesian coordinates (orange), the regularized EM in spherical coordinates (red), and the ESEC algorithm (green). In Fig.~(a), the results calculated by `ESEC' and `rEM\_s' algorithm coincide with each other since they both conserve the velocity magnitude. The dashed line in (b) marks $\mu=1$.}
\label{fig:exploded}
\end{figure}

Fig.~\ref{fig:exploded_v} shows the sample paths of $v(t)$ from four algorithms. For the EM\_C algorithm, the sample path is divergent when $v\to0$. At $t=9.34$, $v$ is smaller than $0.005$ due to the accumulation of truncation errors over many time steps, but one time step later, $v$ jumps to a value larger than $2$. After regularization, the sample path no longer diverges, but the velocity magnitude still deviates from its original value. On the other hand, for the sample paths given by the ESEC algorithm or the EM algorithm in spherical coordinates, the velocity magnitudes remain constant throughout the whole calculation.

Fig.~\ref{fig:exploded_mu} shows the sample paths of $\mu(t)$. By definition, $\mu$ is always less than $1,$ i.e., $|\mu|\leq1$. If the numerical calculation is done in Cartesian coordinates, this condition is automatically satisfied, which is indeed the case for the EM\_C, rEM\_C, and ESEC algorithms in Fig.~\ref{fig:exploded_mu}. However, if the sample path is calculated in the spherical coordinates, this condition is not guaranteed. Due to the intrinsic singularity of the spherical coordinate system, unphysical behaviors may occur at $|\mu|\to1$. In Fig.~\ref{fig:exploded_mu}, the sample path given by rEM\_s deviates from those by the other three algorithms around $t=0.95$ when $\mu$ is close to 1. Then it exceeds 1 and becomes drastically different from the others. It is worth mention that the behavior of $\mu$ calculated in spherical coordinates when $|\mu|\to1$ depends on the implementation and regularization techniques. For example, one can avoid $|\mu|>1$ by directly using the polar angle $\theta$ instead of $\cos\theta$ \citep{lemons2009small,brillinger2012particle}, and the SDE becomes $\mathrm{d}\theta=\frac{1}{2}D_{a}\cot\theta\,\mathrm{d}t+\sqrt{D_{a}}\mathrm{d}W_{\theta}$. However, the singularities at $\theta=0,\pi$ still exist because the coefficient of $\mathrm{d}t$ in the SDE is proportional to $\cot\theta$. Detailed study comparing different implementation and regularization methods in spherical coordinates is beyond the scope of this paper.

\subsection{Electron distribution functions}

From the single-electron trajectories, we see that using the EM method to calculate the pitch-angle scattering brings about a few challenges. In Cartesian coordinates, the EM method does not conserve the energy and the governing differential equation does not satisfied the global Lipschitz condition due to the singularity at $v=0$. If one adopts the EM method in spherical coordinates, then the energy is a constant of motion according to Eq.\,(\ref{eq:Ito_sph_r}), but $\mu=\pm1$ now become singularities for Eq.\,(\ref{eq:Ito_psi}). This issue can be clearly demonstrated when these different algorithms are applied to calculate the distribution function and compared with the analytical solution.

Assume that $\boldsymbol{B}=(0,0,B)$ is in the $\hat{z}$ direction. In the velocity spherical coordinates $(v,\theta,\varphi)$, Eq.~(\ref{FP}) becomes \begin{linenomath} 
\begin{align}
\dfrac{\partial}{\partial t}f_{e}(\boldsymbol{v},t) & =B\dfrac{\partial f_{e}}{\partial\varphi}+\mathcal{L}[f_{e}],\label{FP-spherical}\\
\mathcal{L} & =\dfrac{1}{2v^{3}}\left[\frac{1}{\sin\theta}\frac{\partial}{\partial\theta}\left(\sin\theta\frac{\partial}{\partial\theta}\right)+\frac{1}{\sin^{2}\theta}\frac{\partial^{2}}{\partial\varphi^{2}}\right].
\end{align}
\end{linenomath} The initial condition is $f_{e0}=\delta(\boldsymbol{v}-\boldsymbol{v}_{0})$, $\boldsymbol{v}_{0}=(v_{0},\theta_{0},\varphi_{0})$. Since Eq.\,(\ref{FP-spherical}) is a linear equation in the spherical coordinates, it can be solved in terms of the spherical harmonics \citep{PhysRevE.102.033302}. The solution of the initial value problem is \begin{linenomath} 
\begin{align}
f_{e}(\boldsymbol{v},t)=\dfrac{\delta(v-v_{0})}{v_{0}^{2}\sin\theta_{0}}\sum_{l=0}^{+\infty}\sum_{m=-l}^{l}\overline{Y_{l}^{m}(\theta_{0},\varphi_{0})}Y_{l}^{m}(\theta,\varphi+Bt)\exp[-\dfrac{1}{2v_{0}^{3}}l(l+1)t],\label{eq:analytical_solution}
\end{align}
\end{linenomath} where $\bar{Y}_{l}^{m}$ is the complex conjugate of ${Y}_{l}^{m}$.

From the analytical solution, we can draw several physical insights. Firstly, the coefficients of all the spherical harmonics decay exponentially in time, except for the $l=0$ component. In the limit of $t\to\infty$, $f_{e}$ becomes uniform on the $v=v_{0}$ sphere. Secondly, the magnetic field only influences the dynamics of the azimuthal angle. The azimuthally averaged quantities, such as those of the parallel and perpendicular velocities, evolve in the exactly same way as in the un-magnetized case \citep{PhysRevE.102.033302}.

\begin{figure}[ht]
\centering \includegraphics[width=0.9\textwidth]{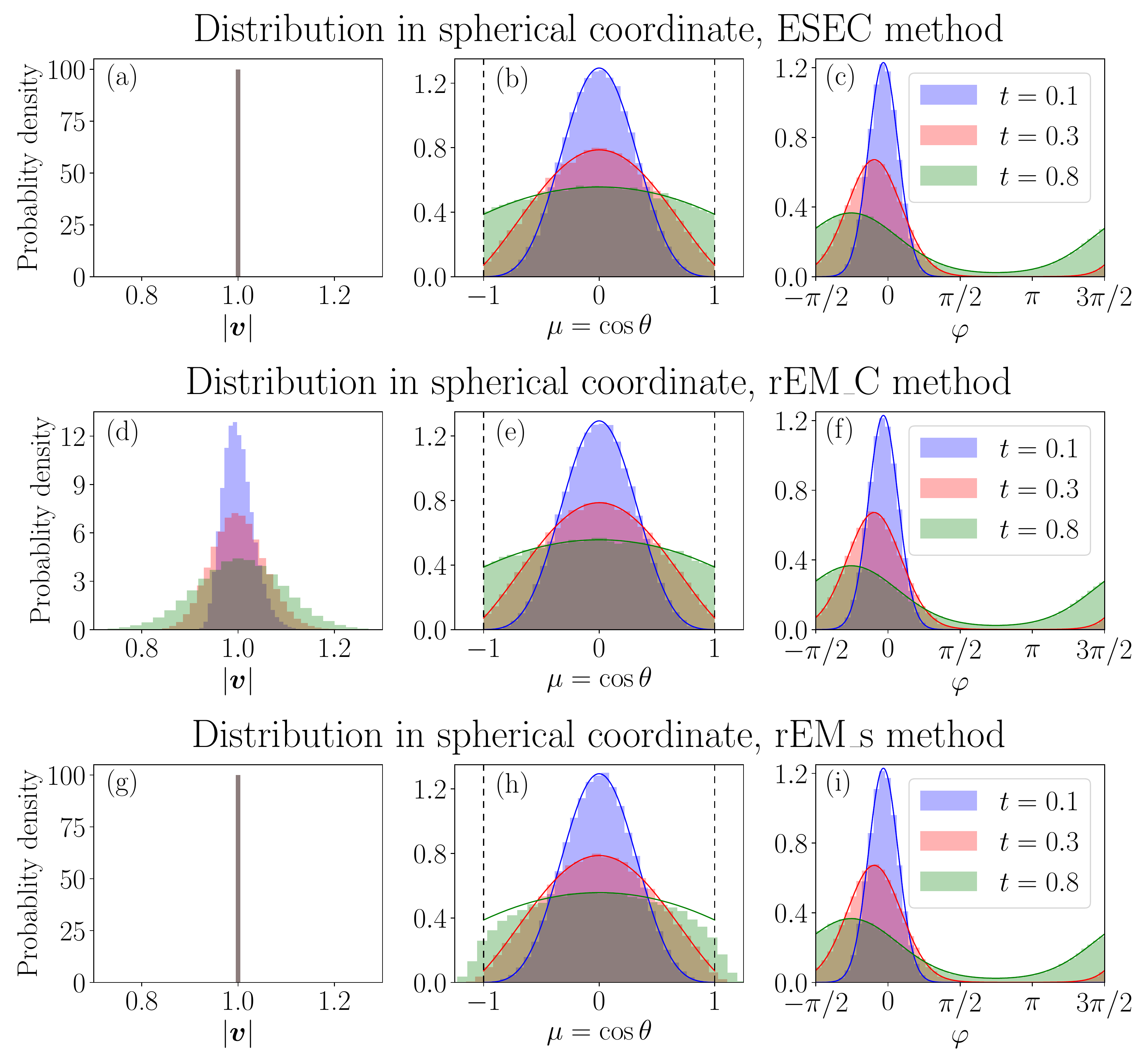} \caption{Distributions in $v$, $\mu=\cos\theta$ and $\varphi$, calculated by the ESEC algorithm, the EM algorithm in Cartesian coordinates (rEM\_C) and spherical coordinates (rEM\_s). The histograms show the distributions obtained numerically by the three algorithms at different times. The solid lines are the analytical solution obtained at each time by integrating Eq.~(\ref{eq:analytical_solution}) and adding up all components with $l\protect\leq50$. The dashed lines mark $\mu=\pm1$.}
\label{fig:theta_phi_distribution}
\end{figure}

We now solve the same problem numerically using the ESEC algorithm and the regularized EM algorithm in both Cartesian and spherical coordinates. Then we benchmark the numerical solution against the analytical solution (\ref{eq:analytical_solution}). As an example, we choose $\boldsymbol{v}_{0}=(1,0,0)$ and $\boldsymbol{B}=(0,0,1)$. In the spherical coordinates, the initial conditions for Eq.\,(\ref{eq:Ito_SDE}) are $v_{0}=1$, $\mu_{0}=0$, and $\varphi_{0}=0$. The distribution function is numerically calculated from $10^{5}$ sample paths with time-step size $h=0.01$. The distributions in spherical coordinates at different times are shown by the histogram in Fig.~\ref{fig:theta_phi_distribution}, while the analytical solutions are shown by the solid lines. The azimuthal angular distributions calculated by all three algorithms are almost identical, while the other two distributions show some interesting differences. The distribution over $\mu$ spreads out from $\mu_{0}=0$ and gradually becomes uniform between $\pm1$. However, in Fig.~\ref{fig:theta_phi_distribution}{\color{blue}h}, the distribution calculated in spherical coordinates gradually leaks out from the boundary $|\mu|=1$, which leads to unphysical results. In comparison, the distributions calculated in Cartesian coordinates do not have this problem. On the other hand, for the distribution in velocity magnitude, the ESEC algorithm and the EM algorithm in spherical coordinates retain the initial delta function at $v=1$ for all time, which is the same as the analytical solution. In contrast, although the distribution calculated by the EM algorithm in Cartesian coordinates is centered around $v=1$, it numerically diffuses over time, which reflects the fact that the EM algorithm in Cartesian coordinates does not conserve energy.

\section{Conclusions and discussions}

In this paper, we have constructed an Explicitly Solvable Energy-Conserving (ESEC) algorithm for the SDE (\ref{eq:Ito_SDE}) describing the pitch-angle scattering process in magnetized plasmas. We have rigorously proved and numerically verified the order $1/2$ strong convergence of the algorithm. The algorithm has also been benchmarked using the analytical solution for the pitch-angle scattering in a constant magnetic field.

In Cartesian coordinates, the coefficients of the Eq.\,(\ref{eq:Ito_SDE}) are singular at $v=0$. In spherical coordinates, the coefficients of the Eq.\,(\ref{eq:Ito_SDE_spherical_coordinate}), (\ref{eq:Ito_psi}) are also singular at $|\mu|=1$. These coefficients do not satisfy the global Lipschitz condition. When standard numerical methods, such as the Euler-Maruyama, are applied, numerical convergence is difficult to establish, and unphysical results may show up. For the proposed ESEC algorithm, the Cartesian-coordinates-based calculation and its energy-preserving property enable us to overcome these obstacles.

We would like to emphasize that pitch-angle scattering is not a physics process that only appears in the zero mass-ratio limit. For collisions between two species with small but finite mass ratio, such as those between electrons and ions in fusion plasmas, the collisional effects will include frictional slowing down and diffusion in the incoming (parallel) direction. Similarly, external electrical field will accelerate particles. In these circumstances, pitch-angle scattering is still an important, if not the most important, component of the collisional process, and the algorithmic difficulty associated with divergence at small $v$ remains. Thus, the pitch-angle scattering component can still be integrated by the ESEC algorithm, which can automatically avoid the numerical divergence at small $v$ due to the perpendicular dynamics. Specifically, the operator-splitting method \citep{higham2002strong} can be applied to separately calculate the pitch-angle scattering and other dynamics at each time step.

Note that when the parallel dynamics is not negligible, the physics of energy conservation is no longer solely represented by the invariance of the velocity norm anymore. As a consequence, the rigorous proof of convergence presented in this paper can not be directly applied. This is a topic for future investigation.

Other future works include the generalization of the ESEC algorithm to more complex and realistic situations. The current paper only considers the particle distribution in the velocity-space. When the background fields are spatially non-uniform, the spatial distribution should also be included. The drift and diffusion coefficients in the SDEs may explicitly depend on the distribution function. For example, the nonlinear Fokker-Plank equation can be solved using methods based on distribution dependent SDEs \citep{barbu2020nonlinear,wang2018distribution,frank2005nonlinear}. Furthermore, the algorithm developed can also be used as a part of the Particle-In-Cell method \citep{allen1994computational}. We plan to use the algorithm to study physics problems involving interaction between particles and self-consistent electromagnetic fields, e.g., the runaway electrons dynamics coupled with plasma instabilities during a tokamak disruption. Another direction for research is to extend the algorithm to high orders. In this regard, the methods used to develop high-order structure-preserving algorithms for both deterministic differential equations \citep{he2015volume,he2016higher,He2016HigherRela,he2017explicit} and SDEs \citep{milstein2002numerical,deng2014high,Wang2016,Hong2017,Zhou2017,Holm2018} could potentially be adapted. 
\begin{acknowledgments}
The authors are grateful to the anonymous referee for the suggestion of the comparison between our algorithm and the EM algorithm in spherical coordinates. X. Zhang would like to thank Francesca Poli and Nicolas Lopez for helpful suggestions. H. Qin thanks Tom Tyranowski and Yajuan Sun for fruitful discussions. This work is supported by US DOE (DE-AC02-09CH11466). 
\end{acknowledgments}

\appendix
%dummy comment inserted by tex2lyx to ensure that this paragraph is not empty%dummy comment inserted by tex2lyx to ensure that this paragraph is not empty%dummy comment inserted by tex2lyx to ensure that this paragraph is not empty%dummy comment inserted by tex2lyx to ensure that this paragraph is not empty%dummy comment inserted by tex2lyx to ensure that this paragraph is not empty

\section{\label{sec:verify_lipschitz} The global Lipschitz condition of coefficients }

A function $f(\boldsymbol{x})$ defined on $\boldsymbol{x}\in\mathbb{R}^{d}$ satisfies the global Lipschitz condition, or is globally Lipschitz continuous, means that for all $\boldsymbol{x}\in\mathbb{R}^{d}$ and $\boldsymbol{y}\in\mathbb{R}^{d}$, there exist a constant $K$ independent of $\boldsymbol{x},\boldsymbol{y}$ such that \begin{linenomath} 
\begin{align}
|f(\boldsymbol{x})-f(\boldsymbol{y})|\leq K|\boldsymbol{x}-\boldsymbol{y}|.\label{eq:lipschitz}
\end{align}
\end{linenomath}

For the coefficients in Eq.~(\ref{eq:Ito_SDE2}), if $\boldsymbol{B}(t)$ is bounded for all $t$, terms $\boldsymbol{v}\times\boldsymbol{B}$, $\boldsymbol{v}$, $\boldsymbol{I}v$ satisfy the global Lipschitz condition uniformly for all $t$. The term $\boldsymbol{vv}/v$ consists of two types of functions, \begin{linenomath} 
\begin{align}
f_{1}(\boldsymbol{x})=\dfrac{x_{1}^{2}}{\sqrt{x_{1}^{2}+x_{2}^{2}+x_{3}^{2}}},\quad f_{2}(\boldsymbol{x})=\dfrac{x_{1}x_{2}}{\sqrt{x_{1}^{2}+x_{2}^{2}+x_{3}^{2}}}.
\end{align}
\end{linenomath} We can verify that both types are globally Lipschitz continuous as follows.

Since $f_{1}$ or $f_{2}$ is differentiable in $\mathbb{R}^{3}/\{0\}$, when $\boldsymbol{x}$ is not proportional to $\boldsymbol{y}$, Eq.~(\ref{eq:lipschitz}) can be proved by the mean value theorem, which states that \begin{linenomath} 
\begin{align}
f(\boldsymbol{x})-f(\boldsymbol{y})=(\boldsymbol{x}-\boldsymbol{y})\cdot\nabla f(\boldsymbol{z}),
\end{align}
\end{linenomath} where $\boldsymbol{z}$ is a point in the segment between $\boldsymbol{x}$ and $\boldsymbol{y}$. By the Cauchy-Schwarz inequality, we have \begin{linenomath} 
\begin{align}
|f(\boldsymbol{x})-f(\boldsymbol{y})|\leq|\nabla f(\boldsymbol{z})|\cdot|\boldsymbol{x}-\boldsymbol{y}|.
\end{align}
\end{linenomath} The gradients of $f_{1}$ and $f_{2}$ are bounded, \begin{linenomath} 
\begin{align}
|\nabla f_{1}|^{2}=2-\dfrac{x_{1}^{4}+2(x_{2}^{2}+x_{3}^{2})^{2}}{(x_{1}^{2}+x_{2}^{2}+x_{3}^{2})^{2}}\leq2,\quad|\nabla f_{2}|^{2}=1-\dfrac{3x_{1}^{2}x_{2}^{2}+(x_{1}^{2}+x_{2}^{2})x_{3}^{2}+x_{3}^{4}}{(x_{1}^{2}+x_{2}^{2}+x_{3}^{2})^{2}}\leq1.
\end{align}
\end{linenomath} Therefore, they satisfy the global Lipschitz condition.

When $\boldsymbol{x}$ is proportional to $\boldsymbol{y}$, the segment between them passes through the undifferentiable point $0$, and the mean value theorem can not be applied directly. In this case, we could write $\boldsymbol{y}=\kappa\boldsymbol{x}$ with some constant $\kappa$, and estimate the difference as \begin{linenomath} 
\begin{align}
|f_{i}(\boldsymbol{x})-f_{i}(\kappa\boldsymbol{x})| & =|(1-|\kappa|)f_{i}(\boldsymbol{x})|=|(1-|\kappa|)x_{i}|\cdot\left|\dfrac{x_{1}}{\sqrt{x_{1}^{2}+x_{2}^{2}+x_{3}^{2}}}\right|\nonumber \\
 & \leq|(1-|\kappa|)\boldsymbol{x}|=\left|\dfrac{1-|\kappa|}{1-\kappa}\right||\boldsymbol{x}-\kappa\boldsymbol{x}|\leq|\boldsymbol{x}-\boldsymbol{y}|,
\end{align}
\end{linenomath} where $i=1,2$.

Therefore, functions $f_{1}$ and $f_{2}$, as well as all coefficients in Eq.~(\ref{eq:Ito_SDE2}), are globally Lipschitz continuous.

\section{\label{sec:appenix} Numerical evaluations of strong and weak errors}

Assuming relative to the analytical solution $\boldsymbol{v}(t_{k};t_{0},\boldsymbol{v}_{0})$, the approximation $\bar{\boldsymbol{v}}(t_{k};t,\boldsymbol{v}_{0},h)$ has order $p_{s}$ strong error and order $p_{w}$ weak error in the entire time interval, then for two different time steps $h_{l}$ and $h_{l+1}$ with $h_{l+1}<h_{l}$, we have \begin{linenomath} 
\begin{align}
\bar{\epsilon}_{\text{s}}^{2}= & \mathbb{E}|\bar{\boldsymbol{v}}(t;t_{0},\boldsymbol{v}_{0},h_{l+1})-\bar{\boldsymbol{v}}(t;t_{0},\boldsymbol{v}_{0},h_{l})|^{2}\nonumber \\
\leq & \mathbb{E}|\boldsymbol{v}(t;t_{0},\boldsymbol{v}_{0})-\bar{\boldsymbol{v}}(t;t_{0},\boldsymbol{v}_{0},h_{l+1})|^{2}+\mathbb{E}|\boldsymbol{v}(t;t_{0},\boldsymbol{v}_{0})-\bar{\boldsymbol{v}}(t;t_{0},\boldsymbol{v}_{0},h_{l})|^{2}\nonumber \\
\leq & (1+\mathbb{E}|\boldsymbol{v}_{0}|^{2})(K_{l}^{2}h_{l}^{2p_{s}}+K_{l+1}^{2}h_{l+1}^{2p_{s}})\leq K(1+\mathbb{E}|\boldsymbol{v}_{0}|^{2})h_{l}^{2p_{s}}.
\end{align}
\end{linenomath} Thus, $\bar{\epsilon}_{\text{s}}$ is order $\mathcal{O}(h_{l}^{p_{s}})$. Similarly for the weak error, we have \begin{linenomath} 
\begin{align}
\bar{\epsilon}_{\text{w}} & =\left|\mathbb{E}\left[\bar{\boldsymbol{v}}(t;t_{0},\boldsymbol{v}_{0},h_{l+1})-\bar{\boldsymbol{v}}(t;t_{0},\boldsymbol{v}_{0},h_{l})\right]\right|\nonumber \\
 & \leq\left|\mathbb{E}\left[{\boldsymbol{v}}(t;t_{0},\boldsymbol{v}_{0})-\bar{\boldsymbol{v}}(t;t_{0},\boldsymbol{v}_{0},h_{l+1})\right]\right|+\left|\mathbb{E}\left[{\boldsymbol{v}}(t;t_{0},\boldsymbol{v}_{0})-\bar{\boldsymbol{v}}(t;t_{0},\boldsymbol{v}_{0},h_{l})\right]\right|\nonumber \\
 & \leq(1+\mathbb{E}|\boldsymbol{v}_{0}|^{2})^{1/2}(K_{l}h_{l}^{p_{w}}+K_{l+1}h_{l+1}^{p_{s}})\leq K(1+\mathbb{E}|\boldsymbol{v}_{0}|^{2})^{1/2}h_{l}^{p_{w}}.
\end{align}
\end{linenomath} Thus, $\bar{\epsilon}_{\text{w}}$ is order $\mathcal{O}(h_{l}^{p_{w}})$.

\section{\label{sec:regularization} Euler-Maruyama scheme in different coordinates}

In this appendix we briefly describe the Euler-Maruyama (EM) scheme for pitch-angle scattering in both Cartesian and spherical coordinate systems. Different regularization techniques used in each coordinate system will be introduced. The method for comparing sample paths with the same underlying Wiener process but calculated in different coordinate systems will also be discussed.

\subsection{\label{sec:rEM_C} EM scheme in Cartesian coordinates }

In Cartesian coordinates, the Ito SDE in Eq.~(\ref{eq:Ito_SDE}) can be re-written as: \begin{linenomath} 
\begin{align}
\mathrm{d}\boldsymbol{v}=\Big(\boldsymbol{v}\times\boldsymbol{B}-\mathcal{F}(v)\hat{\boldsymbol{v}}\Big)\mathrm{d}t+\mathcal{D}(v)(\boldsymbol{I}-\hat{\boldsymbol{v}}\hat{\boldsymbol{v}})\cdot\mathrm{d}\boldsymbol{W},
\end{align}
\end{linenomath} where $\hat{\boldsymbol{v}}=\boldsymbol{v}/v$ is the unit vector parallel to $\boldsymbol{v}$, $\mathcal{F}(v)=D(v)/v$, $\mathcal{D}(v)=\sqrt{D(v)}$, and $D(v)=1/v$. When $v\to0$, the scalar functions $\mathcal{F}$ and $\mathcal{D}$ are divergent, while the rest are finite. Similar to Ref.\,\citep{rosin2014multilevel}, for a chosen critical velocity $v_{c}$, the regularized function $\mathcal{\mathcal{F}}_{\mathrm{r}}(v)$ is defined as \begin{linenomath} 
\begin{align}
\mathcal{F}_{\mathrm{r}}(v):=\begin{cases}
\mathcal{F}(v), & v>v_{c},\\[5pt]
\dfrac{\mathcal{F}^{\prime}\left(v_{c}\right)}{2v_{c}}\left(v^{2}-v_{c}^{2}\right)+\mathcal{F}\left(v_{c}\right), & v\leq v_{c},
\end{cases}
\end{align}
\end{linenomath} where $\mathcal{F}'=\mathrm{d}\mathcal{F}/\mathrm{d}v$. The regularized function $\mathcal{\mathcal{D}}_{\mathrm{r}}(v)$ is similarly defined. It is clear that $\mathcal{D}_{\mathrm{r}}(v)$ and $\mathcal{\mathcal{F}}_{\mathrm{r}}(v)$ are $C^{1}$-continuous and regular at $v=0$. The regularized EM scheme in Cartesian coordiantes (rEM\_C) is \begin{linenomath} 
\begin{align}
\boldsymbol{v}_{k+1}^{\text{rEM\_C}}-\boldsymbol{v}_{k}=\Big(\boldsymbol{v}_{k}\times\boldsymbol{B}_{k}-\mathcal{F}_{\mathrm{r}}(v_{k})\hat{\boldsymbol{v}}_{k}\Big)h+\mathcal{D}_{\mathrm{r}}(v_{k})(\boldsymbol{I}-\hat{\boldsymbol{v}}_{k}\hat{\boldsymbol{v}}_{k})\cdot\Delta\boldsymbol{W}.
\end{align}
\end{linenomath} For the simulation conducted in Sec.\,\ref{sec:constant_B_benchmark}, we set $v_{c}=0.2$.

\subsection{\label{sec:rEM_s} EM scheme in spherical coordinates }

Assuming the magnetic field $B$ is constant and in the $\hat{z}$-direction, the Ito SDE in spherical coordinates can be written as: \begin{linenomath} 
\begin{align}
\mathrm{d}v & =0,\label{eq:Ito_SDE_v}\\
\mathrm{d}\mu & =-D_{a}(v)\mu\,\mathrm{d}t+\sqrt{D_{a}(v)(1-\mu^{2})}\,\mathrm{d}W_{\mu},\label{eq:Ito_SDE_mu}\\
\mathrm{d}\varphi & =-B\,\mathrm{d}t+\sqrt{\dfrac{D_{a}(v)}{1-\mu^{2}}}\,\mathrm{d}W_{\varphi},
\end{align}
\end{linenomath} where $v$ is the velocity magnitude, $\mu=\cos\theta$ is the cosine of the polar angle, $\varphi$ is the azimuthal angle, and $D_{a}(v)=1/v^{3}$. From Eq.~(\ref{eq:Ito_SDE_v}), we know that $v$ and $D_{a}(v)$ are exactly constant. When $|\mu|=1$, the coefficient $\sqrt{1-\mu^{2}}$ is not Lipschitz continuous. In addition, when $|\mu|$ is too close to $1$, i.e., when the particle is too close to north and south pole, the numerical calculation may lead to $|\mu|>1$ and imaginary $\sqrt{1-\mu^{2}}$. Therefore, some regularization methods have to be applied at $|\mu|\to1$.

Here we choose the regularization method used in Ref.~\citep{rosin2014multilevel}, which modified the coefficient $\sqrt{1-\mu^{2}}$ to: \begin{linenomath} 
\begin{align}
\mathcal{M}(\mu)=\begin{cases}
\sqrt{1-\mu^{2}}, & |\mu|<\mu_{c},\\
\sqrt{1-\mu_{c}^{2}}\exp\left[-\left(|\mu|-\mu_{c}\right)S\left(\mu_{c}\right)\right], & |\mu|\geq\mu_{c},
\end{cases}\label{eq:angular_regularization}
\end{align}
\end{linenomath} where $S(\mu_{c})=\mu_{c}/(1-\mu_{c}^{2})$, and $\mu_{c}$ is a chosen critical value. We see that $\mathcal{M}$ is positive and $C^{1}$-continuous for all $\mu\in\mathbb{R}$. To ensure the Einstein relation, the drift term in Eq.~(\ref{eq:Ito_SDE_mu}) is also modified to $D_{a}\mathcal{M}\mathcal{M}'$, where $\mathcal{M}'=\mathrm{d}\mathcal{M}/\mathrm{d}\mu$. The regularized EM scheme in spherical coordinates (rEM\_s) is: \begin{linenomath} 
\begin{align}
\mu_{k+1}^{\text{rEM\_s}}-\mu_{k} & =D_{a}\mathcal{M}(\mu_{k})\mathcal{M}'(\mu_{k})\Delta t+\sqrt{D_{a}}\mathcal{M}(\mu_{k})\Delta W_{\mu},\label{eq:rEM_mu}\\
\varphi_{k+1}^{\text{rEM\_s}}-\varphi_{k} & =B\Delta t+\dfrac{\sqrt{D_{a}}}{\mathcal{M}(\mu_{k})}\Delta W_{\varphi}.
\end{align}
\end{linenomath}

For the simulation in Sec.\,\ref{sec:constant_B_benchmark}, we set $\mu_{c}=0.9$. One thing to notice is that although this regularization method prevents imaginary value, it does not prevent $|\mu|$ from exceeding 1, which is shown in Sec.\,\ref{sec:divergent_sample_path}.

\subsection{\label{sec:generate_weiner} Comparing sample paths in different coordinates }

Usually in the numerical calculations, $\Delta\boldsymbol{W}=(\Delta W_{x},\Delta W_{y},\Delta W_{z})$ or $(\Delta W_{v},\Delta W_{\mu},\Delta W_{\varphi})$ are generated as Gaussian random variables directly. However, if we want to compare sample paths calculated in different coordinates but with the same underlying Wiener process, for example in Sec.\,\ref{sec:divergent_sample_path}, we need to generate $\Delta\boldsymbol{W}$ in one coordinate system, then transform it to the other coordinates.

Assume we have a three dimensional Wiener process $\boldsymbol{W}(t)$ in Cartesian coordinates, but we want to calculate the sample path in spherical coordinates. At each time step $t_{k}$, $\Delta\boldsymbol{W}=\boldsymbol{W}(t_{k+1})-\boldsymbol{W}(t_{k})$ and the sample is at $\boldsymbol{v}_{k}=(v_{k},\mu_{k},\varphi_{k})$. Then $\Delta\boldsymbol{W}$ in spherical coordinates is given by: \begin{linenomath} 
\begin{align}
\begin{pmatrix}\Delta W_{r}\\
\Delta W_{\mu}\\
\Delta W_{\varphi}
\end{pmatrix}=\begin{pmatrix}\sqrt{1-\mu_{k}^{2}}\cos\varphi_{k} & \sqrt{1-\mu_{k}^{2}}\sin\varphi_{k} & \mu_{k}\\
-\mu_{k}\cos\varphi_{k} & -\mu_{k}\sin\varphi_{k} & \sqrt{1-\mu_{k}^{2}}\\
-\sin\varphi_{k} & \cos\varphi_{k} & 0
\end{pmatrix}\begin{pmatrix}\Delta W_{x}\\
\Delta W_{y}\\
\Delta W_{z}
\end{pmatrix}.
\end{align}
\end{linenomath} Here we see the factor $\sqrt{1-\mu^{2}}$ again, which can be regularized using the same method described in the previous section. The inverse transform from spherical to Cartesian coordinates can be performed similarly.

\bibliographystyle{apsrev4-2}
\bibliography{main_article.bib}

\end{document}